\documentclass[review]{elsarticle}

\usepackage{dsfont}
\usepackage{graphicx}
\usepackage{multirow}
\usepackage{stmaryrd}

\usepackage{pstricks}
\usepackage{subfigure}

\usepackage[standard]{ntheorem}
\usepackage{amsfonts,amsmath,amssymb}

\usepackage{alltt}

\usepackage{listings}
\title{Is protein folding problem really a NP-complete one ? First investigations}

\author[ufc]{Jacques M. Bahi}
\ead{jacques.bahi@univ-fcomte.fr}

\author[imag]{Wojciech Bienia}
\ead{wojciech.bienia@g-scop.inpg.fr}

\author[ufc2]{Nathalie C\^{o}t\'{e}}
\ead{nath349@gmail.com}

\author[ufc]{Christophe Guyeux\corref{cor1}\fnref{fn1}}
\ead{christophe.guyeux@univ-fcomte.fr}

\cortext[cor1]{Corresponding author}
\fntext[fn1]{Authors in alphabetic order}

\address[ufc]{FEMTO-ST Institute, UMR 6174 CNRS, University of Franche-Comt\'{e}, Besan\c con, France}

\address[ufc2]{Laboratoire de Biologie du D\'{e}veloppement, UMR 7622, Universit\'{e} Pierre et Marie Curie, Paris, France} 

\address[imag]{G-SCOP Laboratory, ENSIMAG, 46 av. F\'{e}lix Viallet, F-38031 Grenoble Cedex 1, France}

\begin{document}

\begin{abstract}
To determine the 3D conformation of proteins is
a necessity to understand their functions or 
interactions with other molecules. It is commonly
admitted that, when proteins fold from their 
primary linear structures to their final 3D 
conformations, they tend to choose the ones that
minimize their free energy. To find the 3D conformation
of a protein knowing its amino acid sequence,
bioinformaticians use various models of different
resolutions and artificial intelligence tools, as
the protein folding prediction problem is a NP
complete one. More precisely, to determine the
backbone structure of the protein using the 
low resolution models (2D HP square and 3D HP
cubic), by 
finding the conformation that minimize free 
energy, is intractable exactly~\cite{Berger1998}.
 Both the proof
of NP-completeness and the 2D prediction consider
that acceptable conformations have to satisfy a
self-avoiding walk (SAW) requirement, as two different
amino acids cannot occupy a same position in the
lattice. It is shown in this document that the
SAW requirement considered when proving 
NP-completeness is different from the SAW 
requirement used in various prediction programs, and 
that they
are different from the real biological 
requirement. Indeed, the proof of NP completeness 
and the predictions \emph{in silico} consider 
conformations that are not possible in practice.
Consequences of this fact are investigated in this
research work.
\end{abstract}

\begin{keyword}
Protein folding problem \sep Self-Avoiding Walk requirement \sep NP-completeness \sep Graph theory \sep Pivot moves
\end{keyword}

\maketitle

\section{Introduction}

Proteins are polymers formed by different kinds of amino acids.
During or after proteins have been synthesized by ribosomes, they fold to form a specific tridimensional shape.
This 
3D geometric pattern defines their biological functionality, properties, and so on.
For instance, the hemoglobin is able to carry oxygen to the blood
stream thanks to its 3D conformation. 
However, contrary to the
mapping from DNA to the amino acids sequence, the complex folding of
this sequence is not yet understood. 
In fact, Anfinsen's
``Thermodynamic Hypothesis'' claims that the chosen 3D conformation
corresponds to the lowest free energy minimum of the considered
protein~\cite{Anfinsen20071973}.
Efficient constraint programming methods can solve the
problem for reasonably sized sequences~\cite{Dotu}.
But the conformation that minimizes this free energy 
 is most of the time impossible to find in practice,
at least for large proteins, 
due to the number of possible conformations.
Indeed the Protein Structure Prediction (PSP) problem is a
NP-complete one \cite{Crescenzi98,Berger1998}. 
This is why conformations of
proteins are \emph{predicted}: the 3D structures that minimize
the free energy of the protein under consideration are
found by using computational intelligence tools like genetic
algorithms \cite{DBLP:conf/cec/HiggsSHS10}, ant colonies
\cite{Shmygelska2005Feb}, particle swarm
\cite{DBLP:conf/cec/Perez-HernandezRG10}, memetic algorithms
\cite{Islam:2009:NMA:1695134.1695181}, constraint programming~\cite{Dotu,cpsp},
or neural networks
\cite{Dubchak1995}, etc. 
These computational intelligence tools
are coupled with protein energy models (like AMBER,
DISCOVER, or ECEPP/3) to find a conformation that approximately
minimize the free energy of a given protein. 
Furthermore, to face the complexity of the PSP
problem, authors who try to predict the protein folding process use
models of various resolutions. 
For instance, in coarse grain, single-bead models, an amino acid is considered as
a single bead, or point.
These low resolution models are often used as the first stage of the 3D
structure prediction: the backbone of the 3D conformation is
determined. Then, high resolution models come next for further
exploration. Such a prediction strategy is commonly used in PSP
softwares like ROSETTA \cite{Bonneau01,Chivian2005} or TASSER
\cite{Zhang2005}.

In this paper, which is a supplement of \cite{bgc11:ip,bgcs11:ij}, 
we investigate the 2D HP square lattice model. Let us recall that this popular model
is used to test methods and as a first
2D HP lattice folding stage in some protein folding 
prediction algorithms
\cite{DBLP:conf/cec/HiggsSHS10,Braxenthaler97,DBLP:conf/cec/IslamC10,Unger93,DBLP:conf/cec/HorvathC10}.
It focuses only on hydrophobicity by separating
the amino acids into two sets: hydrophobic (H) and hydrophilic (or
polar P) \cite{Dill1985}.
These amino acids occupy vertices of a square lattice, and the 2D low
resolution conformation of the given protein
is thus represented by a self avoiding walk (SAW) on this lattice.
Variations of this model are frequently investigated: 2D or 3D lattices, with square, cubic, triangular, or face-centered-cube shapes.
However, at each time, a SAW requirement for the 
targeted conformation is required.
The PSP problem takes place in that context: given a sequence
of hydrophobic and hydrophilic amino acids, to find the 
self avoiding walk on the lattice that maximizes the number of 
hydrophobic neighbors is a NP complete problem~\cite{Crescenzi98}. 

We will show in this document that this SAW requirement can be understood in various 
different ways, even in the 2D square lattice model. The first understanding of this requirement in the 2D model,
called $SAW_1$ in the remainder of this paper,
has been chosen by authors of~\cite{Crescenzi98} when they have
established the proof of NP-completeness for the PSP problem.
It corresponds to the famous ``excluded volume'' requirement,
and it has been already well-studied by the discrete mathematics
community (see, for instance, the book of Madras and 
Slade~\cite{Madras}). 
It possesses a dynamical formulation we call it $SAW_2$ in
this document.
The $SAW_3$ set is frequently 
chosen by bioinformaticians when they try
to predict the backbone conformation of proteins 
using a low resolution model.
Finally, the last one proposed here 
is perhaps the most realistic one, even
if it still remains far from the biological 
folding operation. We will demonstrate that
these four sets are not equal. In particular,
we will establish that $SAW_4$ is strictly 
included into $SAW_3$, which is strictly included
into $SAW_1=SAW_2$. So the NP-completeness
proof has been realized in a strictly larger set than the
one used for prediction, which is strictly larger than the
set of biologically possible conformations.
Concrete examples of 2D conformations that are
in a $SAW_i$ without being in another $SAW_j$ 
will be given, and characterizations of these
sets, in terms 
of graphs, will finally be proposed.

The remainder of this paper is structured as follows. In the next
section we recall some notations and terminologies on the 2D HP square lattice model, chosen here to simplify explanations.
In Section~\ref{sec:dynamical system}, the 
dynamical system used to describe the folding
process in the 2D model, initially presented
in~\cite{bgc11:ip,bgcs11:ij}, is recalled.
In Sect.~\ref{sec:saw}, various ways to understand the so-called
self-avoiding walk (SAW) requirement are detailed.
Their relations and inclusions are investigated in the next section.
Section~\ref{sec:graph} presents a graph approach to determine the
size ratios between the four SAW requirements defined previously,
and the consequences of their strict inclusions are discussed. 
This paper ends by a conclusion section, in
which our contribution is summarized and intended future work is
presented.

\section{Background}

In the sequel $S^{n}$ denotes the $n^{th}$ term of a sequence $S$ and
$V_{i}$ the $i^{th}$ component of a vector $V$.
 The $k^{th}$
composition of a single function $f$ is represented by $f^{k}=f
\circ...\circ f$.
The set of congruence classes modulo 4 is denoted as $\mathds{Z}/4\mathds{Z}$.
 Finally, given two integers $a<b$, the following notation is used:
$\llbracket a;b\rrbracket =\{a,a+1,\hdots,b\}$.

\subsection{2D Hydrophilic-Hydrophobic (HP) Model}

In the HP model, hydrophobic interactions are supposed to dominate
protein folding.
 This model was formerly introduced by Dill, who
consider in \cite{Dill1985} that the protein core freeing up energy is
formed by hydrophobic amino acids, whereas hydrophilic amino acids
tend to move in the outer surface due to their affinity with the
solvent (see Fig.~\ref{fig:hpmodel}).
 
In this model, a protein conformation is a ``self-avoi\-ding walk (SAW)'', as the walks studied in~\cite{Madras}, on a 2D or 3D lattice such that its energy $E$, depending on topological
neighboring contacts between hydrophobic amino acids that are not
contiguous in the primary structure, is minimal.
In other words, for an amino-acid sequence $P$ of length $\mathsf{N}$ and for the set
$\mathcal{C}(P)$ of all SAW conformations of $P$, the chosen
conformation will be $C^* = min \left\{E(C) \mid C \in \mathcal{C}(P)\right\}$ \cite{Shmygelska05}.
In that context and for a conformation $C$, $E(C)=-q$ where $q$ is equal to the number of
topological hydrophobic neighbors.
 For example, $E(c)=-5$ in
Fig.~\ref{fig:hpmodel}.

\begin{figure}[t]
\centering
\includegraphics[width=2.375in]{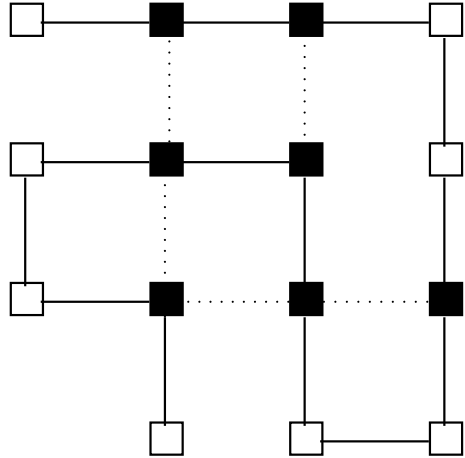}
\caption{Hydrophilic-hydrophobic model (black squares are
hydrophobic residues)}
\label{fig:hpmodel}
\end{figure}

\subsection*{Protein Encoding}

Additionally to the direct coordinate presentation in
the lattice, at least two other
isomorphic encoding strategies for HP models are possible: relative
encoding and absolute encoding. 
In relative encoding \cite{Hoque09},
the move direction is defined relative to the direction of the
previous move (forward, backward, left, or right).
 Alternatively, in absolute encoding
\cite{Backofen99algorithmicapproach}, which is the encoding chosen in
this paper, the direct coordinate presentation is replaced by letters
or numbers representing directions with respect to the lattice
structure.

For absolute encoding in the 2D square lattice, the permitted moves are: east
$\rightarrow$ (denoted by 0), south $\downarrow$ (1), west $\leftarrow$ (2), and north $\uparrow$ (3).
A 2D conformation $C$ of $\mathsf{N}+1$ residues for a protein $P$ is then an element $C$ of $\mathds{Z}/4\mathds{Z}^{\mathsf{N}}$, with a first component equal to 0 (east) \cite{Hoque09}.
For instance, in Fig.~\ref{fig:hpmodel}, the 2D absolute encoding is 00011123322101 (starting from the upper left corner), whereas 001232 corresponds to the following path
in the square lattice: (0,0), (1,0), (2,0), (2,-1), (1,-1), (1,0), (0,0).
In that situation, at most $4^{\mathsf{N}}$ conformations are possible when considering $\mathsf{N}+1$ residues, even if some of them invalidate the SAW requirement as 
defined in~\cite{Madras}.

\section{A Dynamical System for the 2D HP Square Lattice Model}
\label{sec:dynamical system}

Protein minimum energy structure can be considered 
statistically or dynamically. In the latter case, one
speaks in this article of ``protein folding''.
We recall here how to model the folding process in the 2D 
model, or pivot moves, as a dynamical
system. Readers are referred to \cite{bgc11:ip,bgcs11:ij} for further explanations
and to investigate the dynamical behavior of the proteins pivot moves in
this 2D model (it is indeed proven to be chaotic, as
defined by Devaney~\cite{devaney}).

\subsection{Initial Premises}

Let us start with preliminaries introducing some concepts that will be
useful in our approach.

The primary structure of a given protein $P$ with $\mathsf{N}+1$ residues is coded by $0 0 \hdots 0$ ($\mathsf{N}$ times) in absolute encoding.
Its final 2D conformation has an absolute encoding equal to $0 C_1^* \hdots C_{\mathsf{N}-1}^*$, where $\forall i, C_i^* \in \mathds{Z}/4\mathds{Z}$, is such that $E(C^*) = min \left\{E(C) \big/ C \in \mathcal{C}(P)\right\}$.
This final conformation depends on the repartition of hydrophilic and hydrophobic amino acids in the initial sequence.

Moreover, we suppose that, if the residue number $n+1$ is at the east of the residue number $n$ in absolute encoding ($\rightarrow$) and if a fold (pivot move) occurs after $n$, then the east move can only by changed into north ($\uparrow$) or south ($\downarrow$).
That means, in our simplistic model, only rotations
or pivot moves of $+\frac{\pi}{2}$ or $-\frac{\pi}{2}$ are possible.

Consequently, for a given residue that has to be updated, only one of the two possibilities below can appear for its 
absolute encoding during a pivot move:
\begin{itemize}
\item $0 \longmapsto 1$ (that is, east becomes north), $1 \longmapsto 2, 2 \longmapsto 3,$ or $ 3 \longmapsto 0$ 
for a pivot move in the clockwise direction, or
\item $1 \longmapsto 0, 2 \longmapsto 1, 3 \longmapsto 2,$ or $0 \longmapsto 3$ 
for an anticlockwise. 
\end{itemize}

This fact leads to the following definition:
\begin{definition}
The \emph{clockwise fold function} is the function $f: \mathds{Z}/4\mathds{Z} \longrightarrow \mathds{Z}/4\mathds{Z}$ defined by $f(x) = x+1 ~(\textrm{mod}~ 4)$.
\end{definition}
Obviously the anticlockwise fold function is $f^{-1}(x) = x-1 ~(\textrm{mod}~ 4)$.

Thus at the $n^{th}$ folding time, a residue $k$ is chosen and its absolute move is changed by using either $f$ or $f^{-1}$.
As a consequence, \emph{all of the absolute moves must be updated from the coordinate $k$ until the last one $\mathsf{N}$ by using the same folding function}.

\begin{example}
\label{ex1}
If the current conformation is $C=000111$, i.e.,

\begin{center}
\scalebox{0.75} 
{
\begin{pspicture}(0,-2.6)(5.2,2.6)
\psframe[linewidth=0.04,dimen=outer](0.4,2.6)(0.0,2.2)
\psframe[linewidth=0.04,dimen=outer](2.0,2.6)(1.6,2.2)
\psframe[linewidth=0.04,dimen=outer](3.6,2.6)(3.2,2.2)
\psframe[linewidth=0.04,dimen=outer](5.2,2.6)(4.8,2.2)
\psframe[linewidth=0.04,dimen=outer](5.2,1.0)(4.8,0.6)
\psframe[linewidth=0.04,dimen=outer](5.2,-0.6)(4.8,-1.0)
\psframe[linewidth=0.04,dimen=outer](5.2,-2.2)(4.8,-2.6)
\psline[linewidth=0.04cm,arrowsize=0.05291667cm 2.0,arrowlength=1.4,arrowinset=0.4]{->}(0.4,2.4)(1.6,2.4)
\psline[linewidth=0.04cm,arrowsize=0.05291667cm 2.0,arrowlength=1.4,arrowinset=0.4]{->}(2.0,2.4)(3.2,2.4)
\psline[linewidth=0.04cm,arrowsize=0.05291667cm 2.0,arrowlength=1.4,arrowinset=0.4]{->}(3.6,2.4)(4.8,2.4)
\psline[linewidth=0.04cm,arrowsize=0.05291667cm 2.0,arrowlength=1.4,arrowinset=0.4]{->}(5.0,2.2)(5.0,1.0)
\psline[linewidth=0.04cm,arrowsize=0.05291667cm 2.0,arrowlength=1.4,arrowinset=0.4]{->}(5.0,0.6)(5.0,-0.6)
\psline[linewidth=0.04cm,arrowsize=0.05291667cm 2.0,arrowlength=1.4,arrowinset=0.4]{->}(5.0,-1.0)(5.0,-2.2)
\end{pspicture} 
}
\end{center}
 and
if the third residue is chosen to fold (pivot move) by a rotation of $-\frac{\pi}{2}$ (mapping $f$), the new conformation will be
$(C_1,C_2,f(C_3),f(C_4),f(C_5),f(C_6))$, which
is $(0,0,1,2,2,2).$ That is,
\begin{center}
\scalebox{0.75} 
{
\begin{pspicture}(0,-1.0)(5.2,1.0)
\psframe[linewidth=0.04,dimen=outer](2.0,1.0)(1.6,0.6)
\psframe[linewidth=0.04,dimen=outer](3.6,1.0)(3.2,0.6)
\psframe[linewidth=0.04,dimen=outer](5.2,1.0)(4.8,0.6)
\psframe[linewidth=0.04,dimen=outer](2.0,-0.6)(1.6,-1.0)
\psframe[linewidth=0.04,dimen=outer](0.4,-0.6)(0.0,-1.0)
\psframe[linewidth=0.04,dimen=outer](5.2,-0.6)(4.8,-1.0)
\psframe[linewidth=0.04,dimen=outer](3.6,-0.6)(3.2,-1.0)
\psline[linewidth=0.04cm,arrowsize=0.05291667cm 2.0,arrowlength=1.4,arrowinset=0.4]{->}(2.0,0.8)(3.2,0.8)
\psline[linewidth=0.04cm,arrowsize=0.05291667cm 2.0,arrowlength=1.4,arrowinset=0.4]{->}(3.6,0.8)(4.8,0.8)
\psline[linewidth=0.04cm,arrowsize=0.05291667cm 2.0,arrowlength=1.4,arrowinset=0.4]{->}(3.2,-0.8)(2.0,-0.8)
\psline[linewidth=0.04cm,arrowsize=0.05291667cm 2.0,arrowlength=1.4,arrowinset=0.4]{->}(5.0,0.6)(5.0,-0.6)
\psline[linewidth=0.04cm,arrowsize=0.05291667cm 2.0,arrowlength=1.4,arrowinset=0.4]{->}(4.8,-0.8)(3.6,-0.8)
\psline[linewidth=0.04cm,arrowsize=0.05291667cm 2.0,arrowlength=1.4,arrowinset=0.4]{->}(1.6,-0.8)(0.4,-0.8)
\end{pspicture} 
}
\end{center}
\end{example}

These considerations lead to the formalization described thereafter.

\subsection{Formalization and Notations}

Let $\mathsf{N}+1$ be a fixed number of amino acids, where $\mathsf{N}\in\mathds{N}^*= \left\{1,2,3,\hdots\right\}$.
We define 
$$\check{\mathcal{X}}=\mathds{Z}/4\mathds{Z}^\mathsf{N}\times \llbracket -\mathsf{N};\mathsf{N} \rrbracket^\mathds{N}$$ 
as the phase space of all possible folding processes.
An element $X=(C,F)$ of this dynamical folding space is constituted by:
\begin{itemize}
\item A conformation of the $\mathsf{N}+1$ residues in absolute encoding: $C=(C_1,\hdots, C_\mathsf{N}) \in \mathds{Z}/4\mathds{Z}^\mathsf{N}$. Note that we do not require self-avoiding walks here.
\item A sequence $F \in \llbracket -\mathsf{N} ; \mathsf{N} \rrbracket^\mathds{N}$ of future pivot moves such that, when $F_i \in \llbracket -\mathsf{N}; \mathsf{N} \rrbracket$ is $k$, it means that it occurs:
\begin{itemize}
\item a pivot move after the $k-$th residue by a rotation of $-\frac{\pi}{2}$ (mapping $f$) at the $i-$th step, if $k = F_i >0$, 
\item no fold at time $i$ if $k=0$,
\item a pivot move after the $|k|-$th residue by a rotation of $\frac{\pi}{2}$ (\emph{i.e.}, $f^{-1}$) at the $i-$th time, if $k<0$.
\end{itemize}
\end{itemize}
On this phase space, the protein folding dynamic in the 2D model can be formalized as follows.

\medskip

Denote by $i$ the map that transforms a folding sequence in its first term (\emph{i.e.}, in the first folding operation):
$$
\begin{array}{lccl}
i:& \llbracket -\mathsf{N};\mathsf{N} \rrbracket^\mathds{N} & \longrightarrow & \llbracket -\mathsf{N};\mathsf{N} \rrbracket \\
& F & \longmapsto & F^0,
\end{array}$$
by $\sigma$ the shift function over $\llbracket -\mathsf{N};\mathsf{N} \rrbracket^\mathds{N}$, that is to say,
$$
\begin{array}{lccl}
\sigma :& \llbracket -\mathsf{N};\mathsf{N} \rrbracket^\mathds{N} 
 & \longrightarrow & \llbracket -\mathsf{N};\mathsf{N} \rrbracket^\mathds{N} \\
& \left(F^k\right)_{k \in \mathds{N}} & \longmapsto 
 & \left(F^{k+1}\right)_{k \in \mathds{N}},
\end{array}$$
\noindent and by $sign$ the function:
$$
sign(x) = \left\{
\begin{array}{ll}
1 & \textrm{if } x>0,\\
0 & \textrm{if } x=0,\\
-1 & \textrm{else.}
\end{array}
\right.
$$ 
Remark that the shift function removes the first folding operation (a pivot move) from the folding sequence $F$ once it has been achieved.

Consider now the map $G:\check{\mathcal{X}} \to \check{\mathcal{X}}$ defined by:
$$G\left((C,F)\right) = \left( f_{i(F)}(C),\sigma(F)\right),$$
\noindent where $\forall k \in \llbracket -\mathsf{N};\mathsf{N} \rrbracket$, 
$f_k: \mathds{Z}/4\mathds{Z}^\mathsf{N} \to \mathds{Z}/4\mathds{Z}^\mathsf{N}$
is defined by: 
$f_k(C_1,\hdots,C_\mathsf{N}) = (C_1,\hdots,C_{|k|-1},f^{sign(k)}(C_{|k|}),\hdots,f^{sign(k)}(C_\mathsf{N})).$
Thus the folding process of a protein $P$ in the 2D HP square lattice model, with initial conformation equal to $(0,0, \hdots, 0)$ in absolute encoding and a folding sequence equal to $(F^i)_{i \in \mathds{N}}$, is defined by the following dynamical system over $\check{\mathcal{X}}$:
$$
\left\{
\begin{array}{l}
X^0=((0,0,\hdots,0),F)\\
X^{n+1}=G(X^n), \forall n \in \mathds{N}.
\end{array}
\right.
$$

In other words, at each step $n$, if $X^n=(C,F)$, we take the first folding operation to realize, that is $i(F) = F^0 \in \llbracket
-\mathsf{N};\mathsf{N} \rrbracket$, we update the current conformation
$C$ by rotating all of the residues coming after the $|i(F)|-$th one,
which means that we replace the conformation $C$ with $f_{i(F)}(C)$.
Lastly, we remove this rotation (the first term $F^0$) from the folding sequence $F$: $F$ becomes $\sigma(F)$.

\begin{example}
Let us reconsider Example \ref{ex1}.
The unique iteration of this folding process transforms a point of $\check{X}$ having the form $\left((0,0,0,1,1,1);(3, F^1, F^2, \hdots)\right)$ in $G\left((0,0,0,1,1,1),(+3,F^1,F^2, \hdots)\right),$ which is equal to $\left((0,0,1,2,2,2),(F^1,F^2, \hdots)\right)$.
\end{example}

\begin{remark}
Such a formalization allows the study of proteins that never stop to fold, for instance due to never-ending interactions with the environment.
\end{remark}

\begin{remark}
A protein $P$ that has finished to fold, if such a protein exists, has the form $(C,(0,0,0,\hdots))$, where $C$ is the final 2D structure of $P$.
In this case, we can assimilate a folding sequence that is convergent to 0, \emph{i.e.}, of the form $(F^0, \hdots, F^n,0 \hdots)$, with the finite sequence $(F^0, \hdots, F^n)$.
\end{remark}

We will now introduce the SAW requirement in our formulation of the folding process in the 2D model.

\section{The SAW Requirement}
\label{sec:saw}
\subsection{The paths without crossing}
\label{pathWithout}
Let $\mathcal{P}$ denotes the 2D plane,
$$
\begin{array}{cccc}
p: & \mathds{Z}/4\mathds{Z}^\mathsf{N} & \to & \mathcal{P}^{\mathsf{N}+1} \\
 & (C_1, \hdots, C_\mathsf{N}) & \mapsto & (X_0, \hdots, X_\mathsf{N})
\end{array}
$$
where $X_0 = (0,0)$, and
$$
X_{i+1} = \left\{
\begin{array}{ll}
X_i + (1,0) & ~\textrm{if } c_i = 0,\\
X_i + (0,-1) & ~\textrm{if } c_i = 1,\\
X_i + (-1,0) & ~\textrm{if } c_i = 2,\\
X_i + (0,1) & ~\textrm{if } c_i = 3.
\end{array}
\right.
$$

The map $p$ transforms an absolute encoding in its 2D representation.
For instance, $p((0,0,0,1,1,1))$ is ((0,0);(1,0);(2,0);(3,0);(3,-1);(3,-2);(3,-3)), that is, the first figure of Example~\ref{ex1}.

Now, for each $(P_0, \hdots, P_\mathsf{N})$ of $\mathcal{P}^{\mathsf{N}+1}$, we denote by $$support((P_0, \hdots, P_\mathsf{N}))$$ the set (without repetition): $\left\{P_0, \hdots, P_\mathsf{N}\right\}$. For instance,
$$support\left((0,0);(0,1);(0,0);(0,1)\right) = \left\{(0,0);(0,1)\right\}.$$

Then,

\begin{definition}
\label{def:SAW}
A conformation $(C_1, \hdots, C_\mathsf{N}) \in \mathds{Z}/4\mathds{Z}^{\mathsf{N}}$ is a \emph{path without crossing} iff the cardinality of $support(p((C_1, \hdots, C_\mathsf{N})))$ is $\mathsf{N}+1$.
\end{definition}

This path without crossing is sometimes referred as ``excluded
volume'' requirement in the literature. It only means that no
vertex can be occupied by more than one protein monomer.
We can finally remark that Definition \ref{def:SAW} concerns only one conformation, and not a \emph{sequence} of conformations that occurs in a folding process.

\subsection{Defining the SAW Requirements in the 2D model}

The next stage in the formalization of the protein folding process in the 2D model as a dynamical system is to take into account the self-avoiding walk (SAW) requirement, by restricting the set $\mathds{Z}/4\mathds{Z}^\mathsf{N}$ of all possible conformations to one of its subset.
That is, to define precisely the set $\mathcal{C}(P)$ of acceptable conformations of a protein $P$ having $\mathsf{N}+1$ residues.
This stage needs a clear definition of the SAW requirement.
However, as stated above, Definition \ref{def:SAW} only focus on a given conformation, but not on a complete folding process.
In our opinion, this requirement applied to the whole folding process can be understood at least in four ways.

\medskip

In the first and least restrictive approach, we call it ``$SAW_1$'', we only require that the studied conformation satisfies the Definition \ref{def:SAW}. 

\begin{definition}[$SAW_1$]
A conformation $c$ of $\mathds{Z}/4\mathds{Z}^{\mathsf{N}}$ satisfies the first self-avoiding walk requirement ($c \in SAW_1(\mathsf{N})$)
if this conformation is a path without crossing.
\end{definition}

It is not regarded whether this conformation is the result of a folding process that has started from $(0,0,\hdots,0)$.
Such a SAW requirement has been chosen by authors of~\cite{Crescenzi98} when they have proven the NP-completeness of the PSP problem.
It is usually the SAW requirement of biomathematicians, corresponding to the self-avoiding walks studied
in the book of Madras and Slade~\cite{Madras}. 
It is easy to convince ourselves that conformations of $SAW_1$ are the conformations that can be obtained by any chain growth algorithm, like in~\cite{Bornberg-Bauer:1997:CGA:267521.267528}.

As stated before, protein minimum energy structure can be considered statically or dynamically. 
In the latter case, we speak here of ``protein folding'',
since this concerns the dynamic process of folding. When folding on a lattice model, there is an underlying algorithm, 
such as Monte Carlo or genetic algorithm, and an allowed move set. In the following, for the sake of simplicity, only
pivot moves are investigated, but the corner and crankshaft moves should be further investigated~\cite{citeulike:118812}.

Basically, in the protein folding literature, there are methods that require the ``excluded volume'' condition during the
dynamic folding procedure, and those that do not require this condition. 
This is why the second proposed approach called $SAW_2$ requires that, starting from the initial condition $(0,0,\hdots, 0)$, we obtain by a succession of pivot moves a final conformation being a path without crossing.
In other words, we want that the final tree corresponding to the true 2D conformation has 2 vertices with 1 edge and $\mathsf{N}-2$ vertices with 2 edges. 
For instance, the folding process of Figure~\ref{saw2} is acceptable in $SAW_2$, even if it presents a cross in an intermediate conformation.
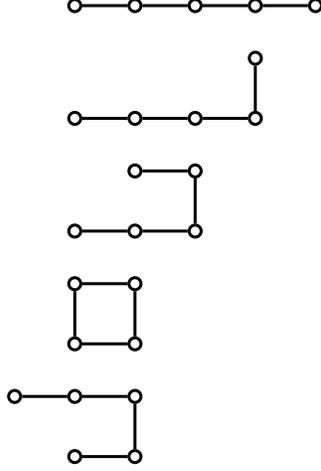
\begin{figure}
\centering
\caption{Folding process acceptable in $SAW_2$ but not in $SAW_3$. The folding sequence (-4,-3,-2,+4), having 3 anticlockwise and 1 clockwise pivot moves,
is applied here on the conformation 0000 represented as the upper line.}
\label{saw2}
\scalebox{1} 
{
\begin{pspicture}(0,-3.1)(4.2,3.1)
\pscircle[linewidth=0.04,dimen=outer](0.9,3.0){0.1}
\pscircle[linewidth=0.04,dimen=outer](1.7,3.0){0.1}
\pscircle[linewidth=0.04,dimen=outer](2.5,3.0){0.1}
\pscircle[linewidth=0.04,dimen=outer](3.3,3.0){0.1}
\pscircle[linewidth=0.04,dimen=outer](4.1,3.0){0.1}
\psline[linewidth=0.04cm](1.0,3.0)(1.6,3.0)
\psline[linewidth=0.04cm](1.8,3.0)(2.4,3.0)
\psline[linewidth=0.04cm](2.6,3.0)(3.2,3.0)
\psline[linewidth=0.04cm](3.4,3.0)(4.0,3.0)
\pscircle[linewidth=0.04,dimen=outer](0.9,1.5){0.1}
\pscircle[linewidth=0.04,dimen=outer](1.7,1.5){0.1}
\pscircle[linewidth=0.04,dimen=outer](2.5,1.5){0.1}
\pscircle[linewidth=0.04,dimen=outer](3.3,1.5){0.1}
\pscircle[linewidth=0.04,dimen=outer](3.3,2.3){0.1}
\psline[linewidth=0.04cm](1.0,1.5)(1.6,1.5)
\psline[linewidth=0.04cm](1.8,1.5)(2.4,1.5)
\psline[linewidth=0.04cm](2.6,1.5)(3.2,1.5)
\psline[linewidth=0.04cm](3.3,1.6)(3.3,2.2)
\pscircle[linewidth=0.04,dimen=outer](0.9,0.0){0.1}
\pscircle[linewidth=0.04,dimen=outer](1.7,0.0){0.1}
\pscircle[linewidth=0.04,dimen=outer](2.5,0.0){0.1}
\psline[linewidth=0.04cm](1.0,0.0)(1.6,0.0)
\psline[linewidth=0.04cm](1.8,0.0)(2.4,0.0)
\pscircle[linewidth=0.04,dimen=outer](2.5,0.8){0.1}
\psline[linewidth=0.04cm](2.5,0.1)(2.5,0.7)
\pscircle[linewidth=0.04,dimen=outer](1.7,0.8){0.1}
\psline[linewidth=0.04cm](1.8,0.8)(2.4,0.8)
\pscircle[linewidth=0.04,dimen=outer](0.9,-1.5){0.1}
\pscircle[linewidth=0.04,dimen=outer](1.7,-1.5){0.1}
\psline[linewidth=0.04cm](1.0,-1.5)(1.6,-1.5)
\pscircle[linewidth=0.04,dimen=outer](1.7,-0.7){0.1}
\psline[linewidth=0.04cm](1.7,-1.4)(1.7,-0.8)
\pscircle[linewidth=0.04,dimen=outer](0.9,-0.7){0.1}
\psline[linewidth=0.04cm](1.0,-0.7)(1.6,-0.7)
\psline[linewidth=0.04cm](0.9,-1.4)(0.9,-0.8)
\pscircle[linewidth=0.04,dimen=outer](0.9,-3.0){0.1}
\pscircle[linewidth=0.04,dimen=outer](1.7,-3.0){0.1}
\psline[linewidth=0.04cm](1.0,-3.0)(1.6,-3.0)
\pscircle[linewidth=0.04,dimen=outer](1.7,-2.2){0.1}
\psline[linewidth=0.04cm](1.7,-2.9)(1.7,-2.3)
\pscircle[linewidth=0.04,dimen=outer](0.9,-2.2){0.1}
\psline[linewidth=0.04cm](1.0,-2.2)(1.6,-2.2)
\pscircle[linewidth=0.04,dimen=outer](0.1,-2.2){0.1}
\psline[linewidth=0.04cm](0.2,-2.2)(0.8,-2.2)
\end{pspicture} 
}
\end{figure}
Such an approach corresponds to programs that start from the initial conformation $(0,0, \hdots, 0)$, fold it several times according to their embedding functions, and then obtain a final conformation on which the SAW property is checked: only the last conformation has to satisfy the Definition \ref{def:SAW}. More precisely,

\begin{definition}[$SAW_2$]
A conformation $c$ of $\mathds{Z}/4\mathds{Z}^{\mathsf{N}}$ satisfies the second self-avoiding walk requirement $SAW_2$ 
if $c \in SAW_1(\mathsf{N})$ and a finite sequence $(F^1,F^2, \hdots,F^n)$ of $\llbracket -\mathsf{N}, \mathsf{N} \rrbracket$ can be found such that $$\left(c,(0,0,\hdots)\right) = G^n\left((0,0,\hdots, 0),\left(F^1,F^2, \hdots, F^n,0, \hdots\right)\right).$$
$SAW_2(\mathsf{N})$ will denote the set of all conformations satisfying this requirement.
\end{definition}

In the next approach, namely the $SAW_3$ requirement, it is demanded that each intermediate conformation, between the initial one and the returned (final) one, satisfies the Definition \ref{def:SAW}.
It restricts the set of all conformations $\mathds{Z}/4\mathds{Z}^\mathsf{N}$, for a given $\mathsf{N}$, to the subset $\mathfrak{C}_\mathsf{N}$ of conformations $(C_1,\hdots, C_\mathsf{N})$ such that $\exists n \in \mathds{N}^*,$ $\exists k_1, \hdots, k_n \in \llbracket -\mathsf{N}; \mathsf{N} \rrbracket$, $$(C_1, \hdots, C_\mathsf{N}) = G^n\left( (0,0, \hdots, 0); (k_1, \hdots, k_n) \right)$$ \emph{and} $\forall i \leqslant n$, the conformation $G^i\left( (0, \hdots, 0); (k_1, \hdots, k_n) \right)$ 
is a path  without crossing.
Let us define it,

\begin{definition}[$SAW_3$]
A conformation $c$ of $\mathds{Z}/4\mathds{Z}^{\mathsf{N}}$ satisfies the third self-avoiding walk requirement 
if $c \in SAW_1(\mathsf{N})$ and a finite sequence $(F^1,F^2, \hdots,F^n)$ of $\llbracket -\mathsf{N}, \mathsf{N} \rrbracket$ can be found such that:
\begin{itemize}
\item $\forall k \in \llbracket 1, n \rrbracket$, the conformation $c_k$ of $G^k\left((0,0,\hdots, 0),\left(F^1,F^2, \hdots, F^n,0, \hdots\right)\right)$ is in $SAW_1(\mathsf{N})$, that is, it is a path without crossing.
\item $\left(c,(0,0,\hdots)\right) = G^n\left((0,0,\hdots, 0),\left(F^1,F^2, \hdots, F^n,0, \hdots\right)\right).$
\end{itemize}
$SAW_3(\mathsf{N})$ will denote the set of all conformations satisfying this requirement.
\end{definition}

The ``SAW requirement'' in the bioinformatics literature 
refers either to the $SAW_2$ or to the $SAW_3$
folding process requirement~\cite{Braxenthaler97,DBLP:conf/cec/HiggsSHS10,DBLP:conf/cec/HorvathC10}. For instance
in \cite{DBLP:conf/cec/IslamC10}, random sequences of $\llbracket 0,3\rrbracket$ are picked and the excluded volume 
requirement (as recalled previously, no vertex can be occupied by more than one protein monomer) is then checked, meaning that this research work 
takes place into $SAW_2$.
Contrarily, in~\cite{Unger93},  the
Monte Carlo search for folding simulations algorithm repeats
the step:
``from conformation $S_i$ with energy $E_i$ make a pivot 
move to get $S_j$ with $E_j$'' until $S_j$ is valid, so
Unger and Moult are in $SAW_3$.
Algorithms that refine progressively their solutions (following
a genetic algorithm or a swarm particle approach for instance)
are often of this kind.
In these $SAW_3$ related approaches, the acceptable conformations are obtained starting from the initial conformation $(0,0, \hdots, 0)$ and are such that all the intermediate conformations satisfy the Definition \ref{def:SAW}. 

Finally, the $SAW_4$ approach is a $SAW_3$ requirement in which there is no intersection of vertex or edge during the transformation of one conformation to another. For instance, the transformation of Figure \ref{saw4} is authorized in the $SAW_3$ approach but refused in the $SAW_4$ one: during the rotation around the residue having a cross, the structure after this residue will intersect the remainder of the ``protein''.
In this last approach it is impossible, for a protein folding from one plane conformation to another plane one, to use the whole space to achieve this folding.
\begin{figure}
\centering
\caption{Folding process acceptable in $SAW_3$ but not in $SAW_4$. It is in $SAW_3$ as $333300111110333333222211111100333$ (the right panel) is $3333001111103333332222111111f^{-1}(1)f^{-1}(1)f^{-1}(0)f^{-1}(0)f^{-1}(0)f^{-1}(0)$, which corresponds
to a clockwise pivot move of residue number 28 in $SAW_3$.
Figure~\ref{pasDansSaw4} explains why this folding process
is not acceptable in $SAW_4$.}
\label{saw4}
\scalebox{0.5} 
{
\begin{pspicture}(0,-6.6)(16.8,6.6)
\psframe[linewidth=0.04,dimen=outer](2.0,3.4)(1.6,3.0)
\psframe[linewidth=0.04,dimen=outer](3.6,5.0)(3.2,4.6)
\psframe[linewidth=0.04,dimen=outer](2.0,1.8)(1.6,1.4)
\psframe[linewidth=0.04,dimen=outer](2.0,5.0)(1.6,4.6)
\psline[linewidth=0.04cm,arrowsize=0.05291667cm 2.0,arrowlength=1.4,arrowinset=0.4]{->}(3.6,4.8)(4.8,4.8)
\psline[linewidth=0.04cm,arrowsize=0.05291667cm 2.0,arrowlength=1.4,arrowinset=0.4]{->}(2.0,4.8)(3.2,4.8)
\psline[linewidth=0.04cm,arrowsize=0.05291667cm 2.0,arrowlength=1.4,arrowinset=0.4]{->}(1.6,6.4)(0.4,6.4)
\psline[linewidth=0.04cm,arrowsize=0.05291667cm 2.0,arrowlength=1.4,arrowinset=0.4]{->}(1.8,1.8)(1.8,3.0)
\psline[linewidth=0.04cm,arrowsize=0.05291667cm 2.0,arrowlength=1.4,arrowinset=0.4]{->}(1.8,3.4)(1.8,4.6)
\psframe[linewidth=0.04,dimen=outer](5.2,5.0)(4.8,4.6)
\psframe[linewidth=0.04,dimen=outer](6.8,5.0)(6.4,4.6)
\psframe[linewidth=0.04,dimen=outer](5.2,3.4)(4.8,3.0)
\psframe[linewidth=0.04,dimen=outer](6.8,3.4)(6.4,3.0)
\psframe[linewidth=0.04,dimen=outer](5.2,1.8)(4.8,1.4)
\psframe[linewidth=0.04,dimen=outer](6.8,1.8)(6.4,1.4)
\psframe[linewidth=0.04,dimen=outer](0.4,-4.6)(0.0,-5.0)
\psframe[linewidth=0.04,dimen=outer](5.2,3.4)(4.8,3.0)
\psframe[linewidth=0.04,dimen=outer](0.4,-6.2)(0.0,-6.6)
\psframe[linewidth=0.04,dimen=outer](5.2,0.2)(4.8,-0.2)
\psframe[linewidth=0.04,dimen=outer](6.8,0.2)(6.4,-0.2)
\psline[linewidth=0.04cm,arrowsize=0.05291667cm 2.0,arrowlength=1.4,arrowinset=0.4]{->}(5.0,4.6)(5.0,3.4)
\psline[linewidth=0.04cm,arrowsize=0.05291667cm 2.0,arrowlength=1.4,arrowinset=0.4]{->}(5.0,3.0)(5.0,1.8)
\psline[linewidth=0.04cm,arrowsize=0.05291667cm 2.0,arrowlength=1.4,arrowinset=0.4]{->}(5.0,1.4)(5.0,0.2)
\psline[linewidth=0.04cm,arrowsize=0.05291667cm 2.0,arrowlength=1.4,arrowinset=0.4]{->}(6.6,0.2)(6.6,1.4)
\psline[linewidth=0.04cm,arrowsize=0.05291667cm 2.0,arrowlength=1.4,arrowinset=0.4]{->}(6.6,1.8)(6.6,3.0)
\psline[linewidth=0.04cm,arrowsize=0.05291667cm 2.0,arrowlength=1.4,arrowinset=0.4]{->}(6.6,3.4)(6.6,4.6)
\psframe[linewidth=0.04,dimen=outer](6.8,6.6)(6.4,6.2)
\psframe[linewidth=0.04,dimen=outer](3.6,6.6)(3.2,6.2)
\psframe[linewidth=0.04,dimen=outer](2.0,6.6)(1.6,6.2)
\psframe[linewidth=0.04,dimen=outer](5.2,6.6)(4.8,6.2)
\psline[linewidth=0.04cm,arrowsize=0.05291667cm 2.0,arrowlength=1.4,arrowinset=0.4]{->}(6.4,6.4)(5.2,6.4)
\psline[linewidth=0.04cm,arrowsize=0.05291667cm 2.0,arrowlength=1.4,arrowinset=0.4]{->}(4.8,6.4)(3.6,6.4)
\psline[linewidth=0.04cm,arrowsize=0.05291667cm 2.0,arrowlength=1.4,arrowinset=0.4]{->}(3.2,6.4)(2.0,6.4)
\psline[linewidth=0.04cm,arrowsize=0.05291667cm 2.0,arrowlength=1.4,arrowinset=0.4]{->}(6.6,5.0)(6.6,6.2)
\psline[linewidth=0.04cm,arrowsize=0.05291667cm 2.0,arrowlength=1.4,arrowinset=0.4]{->}(0.2,-5.0)(0.2,-6.2)
\psline[linewidth=0.04cm,arrowsize=0.05291667cm 2.0,arrowlength=1.4,arrowinset=0.4]{->}(0.2,-3.4)(0.2,-4.6)
\psframe[linewidth=0.04,dimen=outer](0.4,3.4)(0.0,3.0)
\psframe[linewidth=0.04,dimen=outer](0.4,1.8)(0.0,1.4)
\psframe[linewidth=0.04,dimen=outer](0.4,5.0)(0.0,4.6)
\psframe[linewidth=0.04,dimen=outer](0.4,6.6)(0.0,6.2)
\psframe[linewidth=0.04,dimen=outer](0.4,-3.0)(0.0,-3.4)
\psline[linewidth=0.04cm,arrowsize=0.05291667cm 2.0,arrowlength=1.4,arrowinset=0.4]{->}(0.2,6.2)(0.2,5.0)
\psline[linewidth=0.04cm,arrowsize=0.05291667cm 2.0,arrowlength=1.4,arrowinset=0.4]{->}(0.2,4.6)(0.2,3.4)
\psline[linewidth=0.04cm,arrowsize=0.05291667cm 2.0,arrowlength=1.4,arrowinset=0.4]{->}(0.2,3.0)(0.2,1.8)
\psline[linewidth=0.04cm,arrowsize=0.05291667cm 2.0,arrowlength=1.4,arrowinset=0.4]{->}(0.2,1.4)(0.2,0.2)
\psline[linewidth=0.04cm](0.0,-3.0)(0.4,-3.4)
\psline[linewidth=0.04cm](0.0,-3.4)(0.4,-3.0)
\psframe[linewidth=0.04,dimen=outer](12.0,3.4)(11.6,3.0)
\psframe[linewidth=0.04,dimen=outer](13.6,5.0)(13.2,4.6)
\psframe[linewidth=0.04,dimen=outer](12.0,1.8)(11.6,1.4)
\psframe[linewidth=0.04,dimen=outer](12.0,5.0)(11.6,4.6)
\psline[linewidth=0.04cm,arrowsize=0.05291667cm 2.0,arrowlength=1.4,arrowinset=0.4]{->}(13.6,4.8)(14.8,4.8)
\psline[linewidth=0.04cm,arrowsize=0.05291667cm 2.0,arrowlength=1.4,arrowinset=0.4]{->}(12.0,4.8)(13.2,4.8)
\psline[linewidth=0.04cm,arrowsize=0.05291667cm 2.0,arrowlength=1.4,arrowinset=0.4]{->}(11.6,6.4)(10.4,6.4)
\psline[linewidth=0.04cm,arrowsize=0.05291667cm 2.0,arrowlength=1.4,arrowinset=0.4]{->}(11.8,1.8)(11.8,3.0)
\psline[linewidth=0.04cm,arrowsize=0.05291667cm 2.0,arrowlength=1.4,arrowinset=0.4]{->}(11.8,3.4)(11.8,4.6)
\psframe[linewidth=0.04,dimen=outer](13.6,3.4)(13.2,3.0)
\psframe[linewidth=0.04,dimen=outer](15.2,5.0)(14.8,4.6)
\psline[linewidth=0.04cm,arrowsize=0.05291667cm 2.0,arrowlength=1.4,arrowinset=0.4]{->}(15.2,-3.2)(16.4,-3.2)
\psframe[linewidth=0.04,dimen=outer](16.8,5.0)(16.4,4.6)
\psframe[linewidth=0.04,dimen=outer](15.2,3.4)(14.8,3.0)
\psframe[linewidth=0.04,dimen=outer](16.8,3.4)(16.4,3.0)
\psframe[linewidth=0.04,dimen=outer](13.6,1.8)(13.2,1.4)
\psframe[linewidth=0.04,dimen=outer](15.2,1.8)(14.8,1.4)
\psframe[linewidth=0.04,dimen=outer](16.8,1.8)(16.4,1.4)
\psframe[linewidth=0.04,dimen=outer](12.0,-3.0)(11.6,-3.4)
\psframe[linewidth=0.04,dimen=outer](15.2,3.4)(14.8,3.0)
\psframe[linewidth=0.04,dimen=outer](13.6,-3.0)(13.2,-3.4)
\psframe[linewidth=0.04,dimen=outer](15.2,-3.0)(14.8,-3.4)
\psframe[linewidth=0.04,dimen=outer](16.8,-3.0)(16.4,-3.4)
\psline[linewidth=0.04cm,arrowsize=0.05291667cm 2.0,arrowlength=1.4,arrowinset=0.4]{->}(15.0,4.6)(15.0,3.4)
\psline[linewidth=0.04cm,arrowsize=0.05291667cm 2.0,arrowlength=1.4,arrowinset=0.4]{->}(15.0,3.0)(15.0,1.8)
\psline[linewidth=0.04cm,arrowsize=0.05291667cm 2.0,arrowlength=1.4,arrowinset=0.4]{->}(15.0,1.4)(15.0,0.2)
\psline[linewidth=0.04cm,arrowsize=0.05291667cm 2.0,arrowlength=1.4,arrowinset=0.4]{->}(13.4,1.8)(13.4,3.0)
\psline[linewidth=0.04cm,arrowsize=0.05291667cm 2.0,arrowlength=1.4,arrowinset=0.4]{->}(16.6,0.2)(16.6,1.4)
\psline[linewidth=0.04cm,arrowsize=0.05291667cm 2.0,arrowlength=1.4,arrowinset=0.4]{->}(16.6,1.8)(16.6,3.0)
\psline[linewidth=0.04cm,arrowsize=0.05291667cm 2.0,arrowlength=1.4,arrowinset=0.4]{->}(16.6,3.4)(16.6,4.6)
\psframe[linewidth=0.04,dimen=outer](16.8,6.6)(16.4,6.2)
\psframe[linewidth=0.04,dimen=outer](13.6,6.6)(13.2,6.2)
\psframe[linewidth=0.04,dimen=outer](12.0,6.6)(11.6,6.2)
\psframe[linewidth=0.04,dimen=outer](15.2,6.6)(14.8,6.2)
\psline[linewidth=0.04cm,arrowsize=0.05291667cm 2.0,arrowlength=1.4,arrowinset=0.4]{->}(16.4,6.4)(15.2,6.4)
\psline[linewidth=0.04cm,arrowsize=0.05291667cm 2.0,arrowlength=1.4,arrowinset=0.4]{->}(14.8,6.4)(13.6,6.4)
\psline[linewidth=0.04cm,arrowsize=0.05291667cm 2.0,arrowlength=1.4,arrowinset=0.4]{->}(13.2,6.4)(12.0,6.4)
\psline[linewidth=0.04cm,arrowsize=0.05291667cm 2.0,arrowlength=1.4,arrowinset=0.4]{->}(16.6,5.0)(16.6,6.2)
\psline[linewidth=0.04cm,arrowsize=0.05291667cm 2.0,arrowlength=1.4,arrowinset=0.4]{->}(12.0,-3.2)(13.2,-3.2)
\psline[linewidth=0.04cm,arrowsize=0.05291667cm 2.0,arrowlength=1.4,arrowinset=0.4]{->}(10.4,-3.2)(11.6,-3.2)
\psframe[linewidth=0.04,dimen=outer](10.4,3.4)(10.0,3.0)
\psframe[linewidth=0.04,dimen=outer](10.4,1.8)(10.0,1.4)
\psframe[linewidth=0.04,dimen=outer](10.4,5.0)(10.0,4.6)
\psframe[linewidth=0.04,dimen=outer](10.4,6.6)(10.0,6.2)
\psframe[linewidth=0.04,dimen=outer](10.4,-3.0)(10.0,-3.4)
\psline[linewidth=0.04cm,arrowsize=0.05291667cm 2.0,arrowlength=1.4,arrowinset=0.4]{->}(10.2,6.2)(10.2,5.0)
\psline[linewidth=0.04cm,arrowsize=0.05291667cm 2.0,arrowlength=1.4,arrowinset=0.4]{->}(10.2,4.6)(10.2,3.4)
\psline[linewidth=0.04cm,arrowsize=0.05291667cm 2.0,arrowlength=1.4,arrowinset=0.4]{->}(10.2,3.0)(10.2,1.8)
\psline[linewidth=0.04cm](10.0,-3.0)(10.4,-3.4)
\psline[linewidth=0.04cm](10.0,-3.4)(10.4,-3.0)
\psline[linewidth=0.04cm,arrowsize=0.05291667cm 2.0,arrowlength=1.4,arrowinset=0.4]{->}(0.4,-6.4)(1.6,-6.4)
\psframe[linewidth=0.04,dimen=outer](2.0,-6.2)(1.6,-6.6)
\psframe[linewidth=0.04,dimen=outer](3.6,-6.2)(3.2,-6.6)
\psline[linewidth=0.04cm,arrowsize=0.05291667cm 2.0,arrowlength=1.4,arrowinset=0.4]{->}(2.0,-6.4)(3.2,-6.4)
\psframe[linewidth=0.04,dimen=outer](12.0,0.2)(11.6,-0.2)
\psline[linewidth=0.04cm,arrowsize=0.05291667cm 2.0,arrowlength=1.4,arrowinset=0.4]{->}(11.8,0.2)(11.8,1.4)
\psframe[linewidth=0.04,dimen=outer](13.6,0.2)(13.2,-0.2)
\psline[linewidth=0.04cm,arrowsize=0.05291667cm 2.0,arrowlength=1.4,arrowinset=0.4]{->}(13.4,0.2)(13.4,1.4)
\psframe[linewidth=0.04,dimen=outer](10.4,0.2)(10.0,-0.2)
\psline[linewidth=0.04cm,arrowsize=0.05291667cm 2.0,arrowlength=1.4,arrowinset=0.4]{->}(10.2,1.4)(10.2,0.2)
\psframe[linewidth=0.04,dimen=outer](12.0,-1.4)(11.6,-1.8)
\psline[linewidth=0.04cm,arrowsize=0.05291667cm 2.0,arrowlength=1.4,arrowinset=0.4]{->}(11.8,-1.4)(11.8,-0.2)
\psframe[linewidth=0.04,dimen=outer](13.6,-1.4)(13.2,-1.8)
\psline[linewidth=0.04cm,arrowsize=0.05291667cm 2.0,arrowlength=1.4,arrowinset=0.4]{->}(13.4,-1.4)(13.4,-0.2)
\psframe[linewidth=0.04,dimen=outer](10.4,-1.4)(10.0,-1.8)
\psline[linewidth=0.04cm,arrowsize=0.05291667cm 2.0,arrowlength=1.4,arrowinset=0.4]{->}(10.2,-0.2)(10.2,-1.4)
\psline[linewidth=0.04cm,arrowsize=0.05291667cm 2.0,arrowlength=1.4,arrowinset=0.4]{->}(13.4,-3.0)(13.4,-1.8)
\psline[linewidth=0.04cm,arrowsize=0.05291667cm 2.0,arrowlength=1.4,arrowinset=0.4]{->}(10.2,-1.8)(10.2,-3.0)
\psframe[linewidth=0.04,dimen=outer](15.2,0.2)(14.8,-0.2)
\psframe[linewidth=0.04,dimen=outer](16.8,0.2)(16.4,-0.2)
\psframe[linewidth=0.04,dimen=outer](15.2,-1.4)(14.8,-1.8)
\psframe[linewidth=0.04,dimen=outer](16.8,-1.4)(16.4,-1.8)
\psframe[linewidth=0.04,dimen=outer](15.2,0.2)(14.8,-0.2)
\psline[linewidth=0.04cm,arrowsize=0.05291667cm 2.0,arrowlength=1.4,arrowinset=0.4]{->}(15.0,-1.8)(15.0,-3.0)
\psline[linewidth=0.04cm,arrowsize=0.05291667cm 2.0,arrowlength=1.4,arrowinset=0.4]{->}(15.0,-0.2)(15.0,-1.4)
\psline[linewidth=0.04cm,arrowsize=0.05291667cm 2.0,arrowlength=1.4,arrowinset=0.4]{->}(16.6,-3.0)(16.6,-1.8)
\psline[linewidth=0.04cm,arrowsize=0.05291667cm 2.0,arrowlength=1.4,arrowinset=0.4]{->}(16.6,-1.4)(16.6,-0.2)
\psline[linewidth=0.04cm,arrowsize=0.05291667cm 2.0,arrowlength=1.4,arrowinset=0.4]{->}(5.2,-3.2)(6.4,-3.2)
\psframe[linewidth=0.04,dimen=outer](5.2,-3.0)(4.8,-3.4)
\psframe[linewidth=0.04,dimen=outer](6.8,-3.0)(6.4,-3.4)
\psframe[linewidth=0.04,dimen=outer](5.2,-1.4)(4.8,-1.8)
\psframe[linewidth=0.04,dimen=outer](6.8,-1.4)(6.4,-1.8)
\psline[linewidth=0.04cm,arrowsize=0.05291667cm 2.0,arrowlength=1.4,arrowinset=0.4]{->}(5.0,-1.8)(5.0,-3.0)
\psline[linewidth=0.04cm,arrowsize=0.05291667cm 2.0,arrowlength=1.4,arrowinset=0.4]{->}(5.0,-0.2)(5.0,-1.4)
\psline[linewidth=0.04cm,arrowsize=0.05291667cm 2.0,arrowlength=1.4,arrowinset=0.4]{->}(6.6,-3.0)(6.6,-1.8)
\psline[linewidth=0.04cm,arrowsize=0.05291667cm 2.0,arrowlength=1.4,arrowinset=0.4]{->}(6.6,-1.4)(6.6,-0.2)
\psframe[linewidth=0.04,dimen=outer](2.0,0.2)(1.6,-0.2)
\psline[linewidth=0.04cm,arrowsize=0.05291667cm 2.0,arrowlength=1.4,arrowinset=0.4]{->}(1.8,0.2)(1.8,1.4)
\psframe[linewidth=0.04,dimen=outer](2.0,-1.4)(1.6,-1.8)
\psline[linewidth=0.04cm,arrowsize=0.05291667cm 2.0,arrowlength=1.4,arrowinset=0.4]{->}(1.8,-1.4)(1.8,-0.2)
\psframe[linewidth=0.04,dimen=outer](0.4,0.2)(0.0,-0.2)
\psframe[linewidth=0.04,dimen=outer](0.4,-1.4)(0.0,-1.8)
\psline[linewidth=0.04cm,arrowsize=0.05291667cm 2.0,arrowlength=1.4,arrowinset=0.4]{->}(0.2,-0.2)(0.2,-1.4)
\psline[linewidth=0.04cm,arrowsize=0.05291667cm 2.0,arrowlength=1.4,arrowinset=0.4]{->}(0.2,-1.8)(0.2,-3.0)
\psline[linewidth=0.04cm,arrowsize=0.05291667cm 2.0,arrowlength=1.4,arrowinset=0.4]{->}(3.6,-6.4)(4.8,-6.4)
\psframe[linewidth=0.04,dimen=outer](5.2,-6.2)(4.8,-6.6)
\psframe[linewidth=0.04,dimen=outer](6.8,-6.2)(6.4,-6.6)
\psline[linewidth=0.04cm,arrowsize=0.05291667cm 2.0,arrowlength=1.4,arrowinset=0.4]{->}(5.2,-6.4)(6.4,-6.4)
\end{pspicture} 
}
\end{figure}

\begin{figure}
\caption{It is impossible to make the rotation around the crossed square, in such a way that the tail does not intersect the head structure during the rotation, so the folding process of Fig.~\ref{saw4} is not in $SAW_4$.}
\label{pasDansSaw4}
\centering
\includegraphics[scale=0.45]{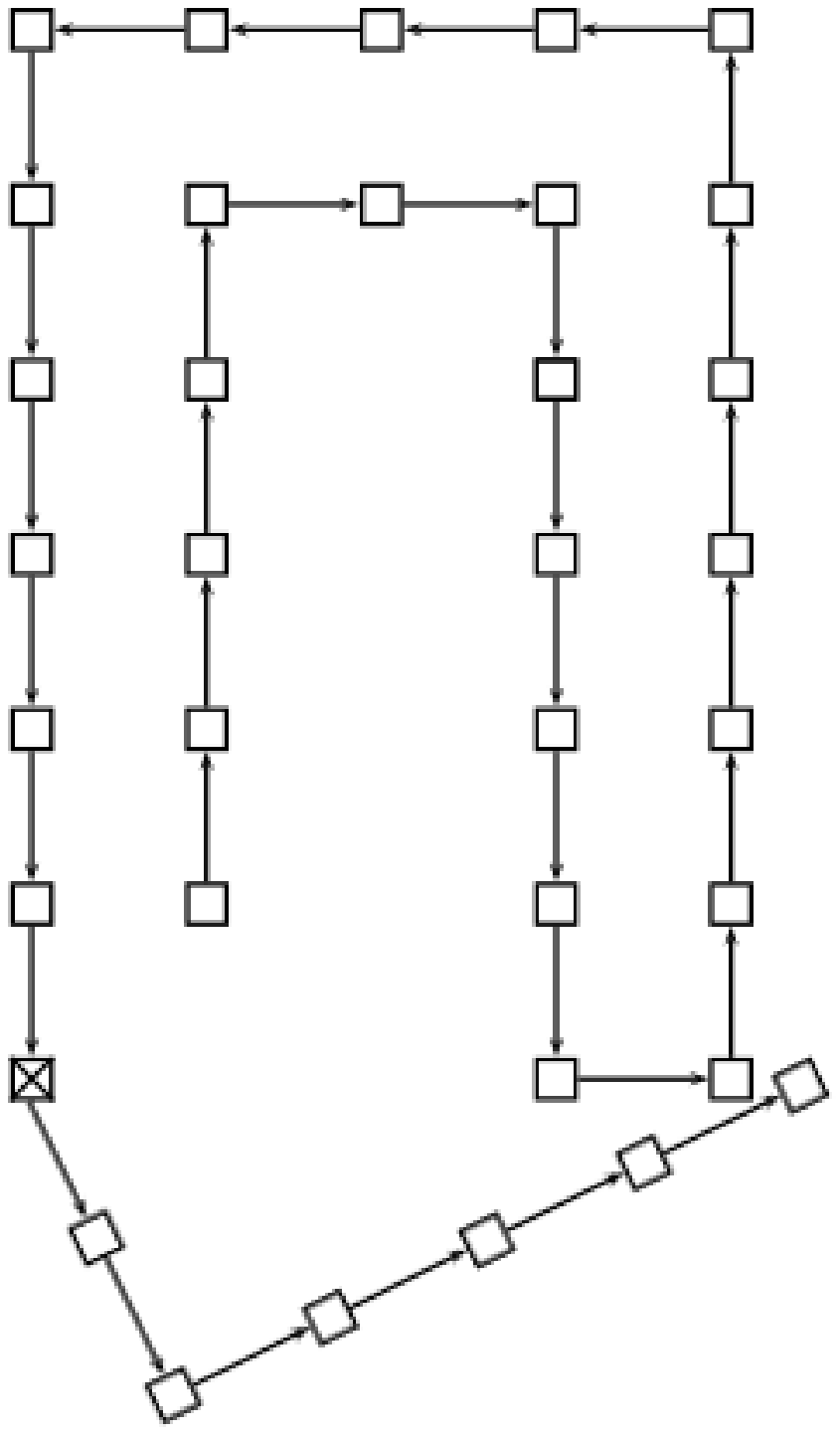}
\end{figure}

This last requirement is the closest approach of a true natural protein folding. It is related to researches that consider
more complicated moves than the simple pivot move~\cite{citeulike:118812}.

\section{Relations between the SAW requirements}

For $i \in \{1,2,3,4\}$, the set $\displaystyle{\bigcup_{n \in \mathds{N}^*}SAW_i(n)}$ will be simply written $SAW_i$.
The following inclusions hold obviously: $$SAW_4 \subseteq SAW_3 \subseteq SAW_2 \subseteq SAW_1$$ due to the definitions of the SAW requirements presented in the previous section.
Additionally, Figure \ref{saw4} shows that $SAW_4 \neq SAW_3$, thus we have,

\begin{proposition}
\label{subsets}
$SAW_4 \subsetneq SAW_3 \subseteq SAW_2 \subseteq SAW_1$. 
\end{proposition}

Let us investigate more precisely the links between $SAW_1, SAW_2$, and $SAW_3$.

\subsection{$SAW_1$ is $SAW_2$}

Let us now prove that,

\begin{proposition}
$\forall n \in \mathds{N}, SAW_1(n)=SAW_2(n)$.
\end{proposition}

\begin{proof}
We need to prove that $SAW_1(n) \subset SAW_2(n)$, \emph{i.e.}, that 
any conformation of $SAW_1(n)$ can be obtained from $(0,0,..,0)$ by
operating a sequence of (anti)clockwise pivot moves. 

Obviously, to start from the conformation $(0,0,..,0)$ is equivalent than
to start with the conformation $(c,c,...,c)$, where $c \in \{0,1,2,3\}$.
Thus the initial configuration is characterized by the absence of a change
in the values (the initial sequence is a constant one).

We will now prove the result by a mathematical induction on the
number $k$ of changes in the sequence.
\begin{itemize}
\item The base case is obvious, as the 4 conformations with no change are in $SAW_1(n)\cap SAW_2(n)$.
\item Let us suppose that the statement holds for some $k \geqslant 1$. Let 
$c=(c_1,c_2,...,c_n)$ be a conformation having exactly $k$ changes, that is, the cardinality of
the set $D(c) = \left\{i \in \llbracket 1, n-1 \rrbracket \big/ c_{i+1} \neq c_i \right\}$ is $k$.
Let us denote by $p(c)$ the first change in this sequence: $p(c) = min \left\{D(c)\right\}$.
We can apply the folding operation that suppress the difference between $c_{p(c)}$ and
$c_{p(c)+1}$. For instance, if $c_{p(c)+1} = c_{p(c)}-1 ~(mod~4)$, then a clockwise pivot move
on position $c_{p(c)+1}$ will remove this difference. So the conformation
$c'=\left(c_1,c_2,\hdots,c_{p(c)},f\left(c_{p(c)+1}\right),\hdots,f(c_n)\right)$ has $k-1$
changes. By the induction hypothesis, $c'$ can be obtained from $(j,j,j,\hdots,j)$, where 
$j \in \{0,1,2,3\}$ by a succession of clockwise and anticlockwise pivot move. We can conclude
that it is the case for $c$ too.
\end{itemize}
\end{proof}

Indeed the notion of ``pivot moves'' is well-known in
the literature on protein folding. It was already supposed 
that pivot moves provide an ergodic move set, meaning that by 
a sequence of pivot moves one can transform any conformation
into any other conformation, when only requiring that the
ending conformation satisfies the excluded volume requirement.
The contribution of this section is simply a rigorous proof 
of such an assumption.

\subsection{$SAW_2$ is not $SAW_3$}
\label{Saw2pasSaw3}

To determine whether $SAW_2$ is equal to $SAW_3$, we have firstly followed a computational 
approach, by using the Python language. A first function (the function \emph{conformations} 
of Listing~\ref{AllConformations} in the appendix)
has been written to return the list of all possible conformations, even if they are not
paths without crossing. In other words, this function produces 
all sequences of compass directions of length $n$ (thus $conformations(n) = \mathds{Z}/4\mathds{Z}^n$).
Then a generator \emph{saw1\_conformations(n)} has been constructed, making
it possible to obtain all the $SAW_1$ conformations (see Listing~\ref{SAW1Conformations}). 
It is based on the fact that such a conformation
of length $n$ must have a support of cardinality equal to $n$.

Finally, a program (Algorithm~\ref{SAW3Conformations}) has been written to check experimentally 
whether an element of $SAW_1=SAW_2$ is in $SAW3$. This is a systematic approach: for
each residue of the candidate conformation, we try to make a clockwise pivot move and an
anticlockwise one. If the obtained conformation is a path without crossing then 
the candidate is rejected. On the contrary, if it is never possible to unfold the
protein, whatever the considered residue, then the candidate is in $SAW_2$ without being
in $SAW_3$.

Figure~\ref{undefoldable} gives four examples of conformations that are
in $SAW_2$ without being in $SAW_3$ (the unique ones authors have found
via the programs given in the appendix). These counterexamples prove that,

\begin{proposition}
$\exists n \in \mathds{N}^*, SAW_2(n) \neq SAW_3(n).$
\end{proposition}

\begin{figure}[h!]
 \centering
 \subfigure[175 nodes]{\includegraphics[width=0.24\textwidth]
 {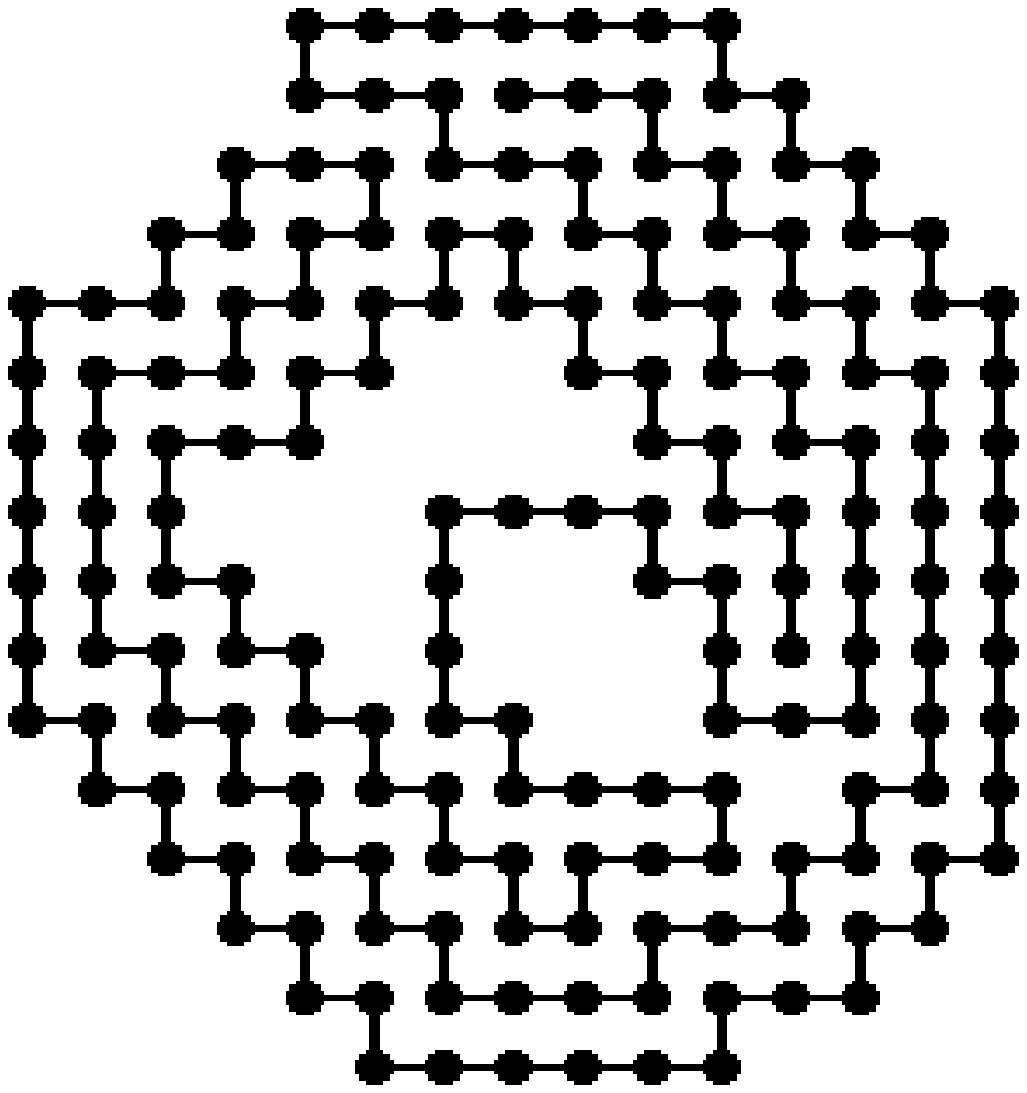}\label{cs175}}\hspace{2cm}
 \subfigure[159 nodes]{\includegraphics[width=0.24\textwidth]
 {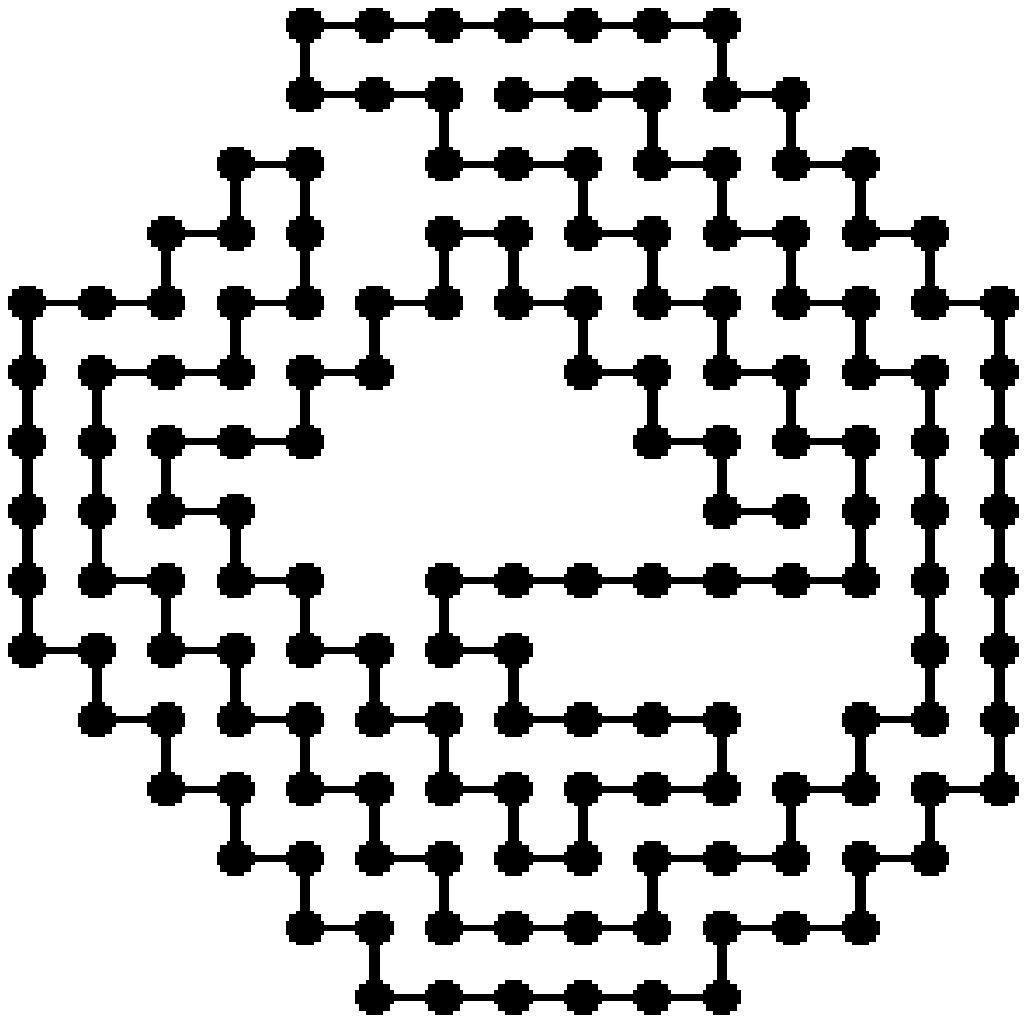}\label{cs159}}\\
 \subfigure[169 nodes]{\includegraphics[width=0.24\textwidth]
 {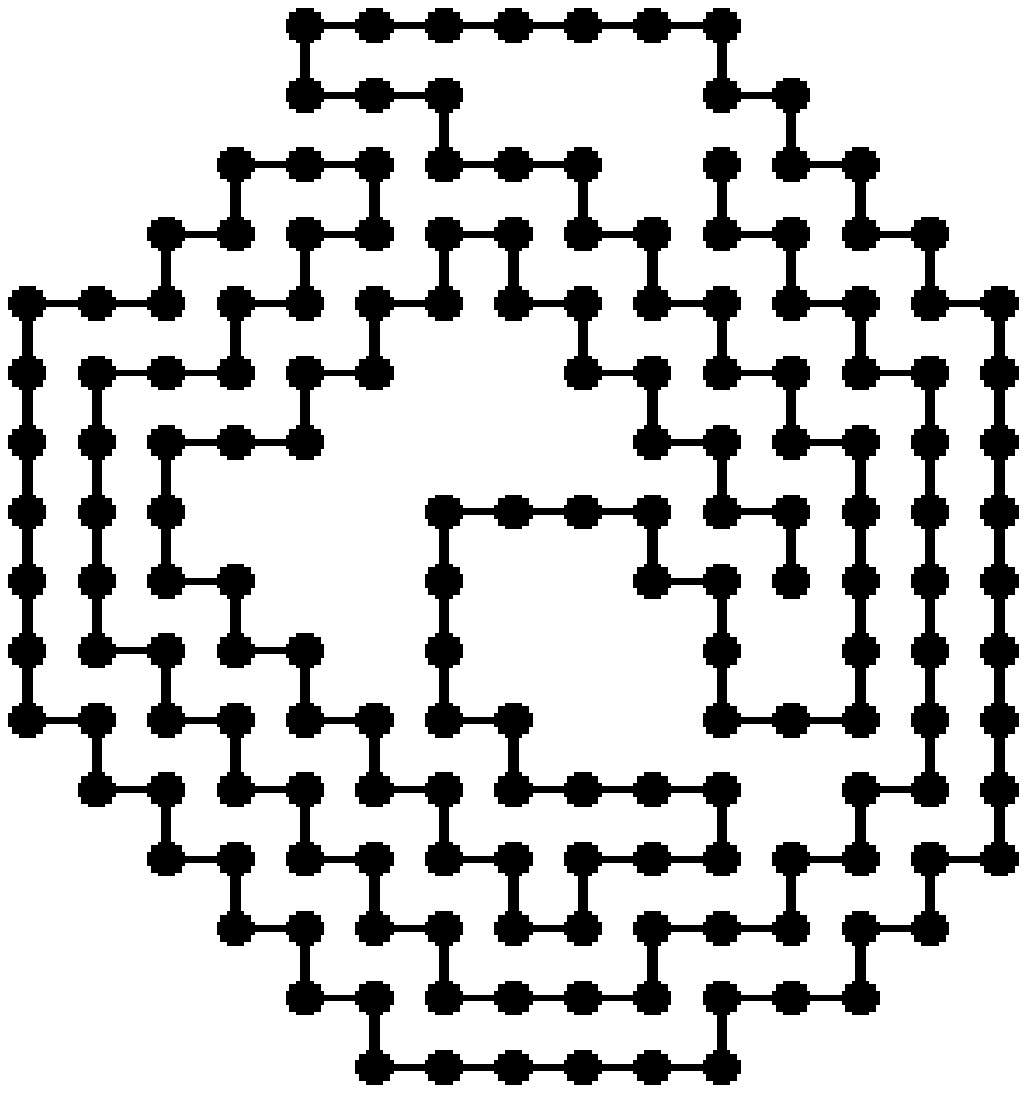}\label{cs169}}\hspace{2cm}
 \subfigure[914 nodes]{
\includegraphics[width=0.5\textwidth]{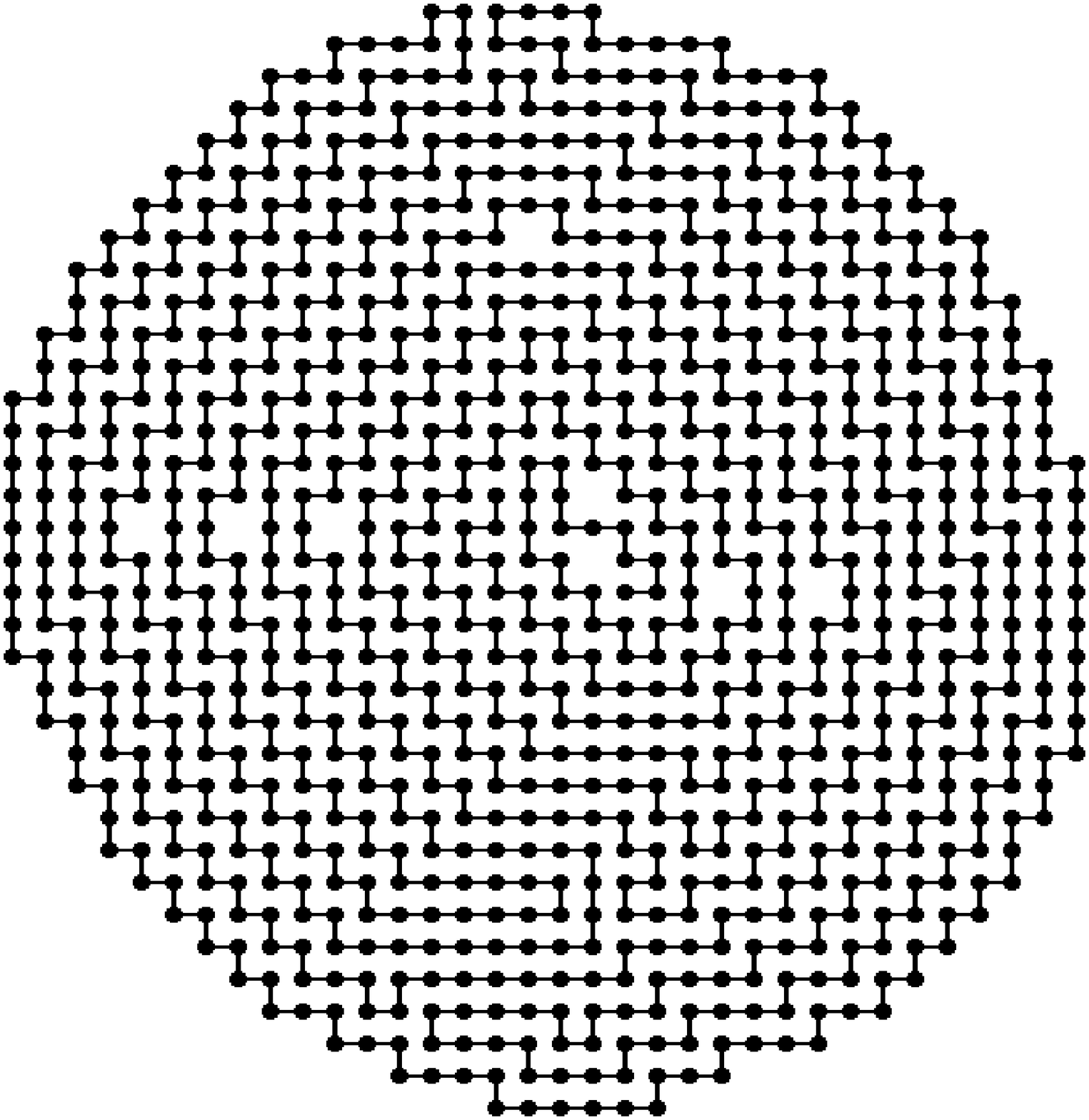}\label{grosCS}}
\caption{Examples of conformations in $SAW_2$ without being in $SAW_3$}
\label{undefoldable}
\end{figure}

\subsubsection{Consequences of the strict inclusion}

Proposition~\ref{subsets} can be rewritten as follows,

\begin{proposition}
$SAW_4 \subsetneq SAW_3 \subsetneq SAW_2 = SAW_1$. 
\end{proposition}

As stated previously, the NP-completeness holds for $SAW_1$. However $SAW_1$
is a strictly larger set than $SAW_3$. $SAW_3$ is a set frequently used for protein
structure prediction. As $SAW_3$ is strictly smaller than $SAW_1$, it is not
sure that the considered problem still remains a NP complete one. Incidentally,
it is not clear that only prediction is possible. Indeed, proteins have ``only'' 
tens to thousands amino acids. If $SAW_3$ is very small 
compared to $SAW_1$, then perhaps exact methods as SAT solvers can be more widely
considered ?

Moreover, $SAW_3$ is strictly larger than $SAW_4$, which is a  
2D model slightly closer than true real protein folding. This strict inclusion reinforces the
fact that the NP-completeness statement must be regarded another time, to 
determine if this prediction problem is indeed a NP-complete one or not. 
Furthermore, prediction tools could reduce the set of all possibilities by
taking place into $SAW_4$ instead of $SAW_3$, thus
improving the confidence put in the returned conformations.

All of these questionings are strongly linked to the size ratio between each
$SAW_i$: the probability the NP-completeness proof remains valid in $SAW_3$ or
$SAW_4$ decreases when these ratios increase. 
This is why we will investigate more deeply, in the next section, the relation 
between $SAW_2$ and $SAW_3$

\section{A Graph Approach of the $SAW_i$ Ratios Problem}
\label{sec:graph}

Let us denote by $\mathfrak{G}_0(n)$ the directed graph having $4^n$ vertices, such that:
\begin{itemize}
\item these vertices are elements of $\llbracket 0,3 \rrbracket^n$,
\item there is a directed edge from the vertex $(e_1, \hdots, e_n)$ to the vertex $(f_1,\hdots,f_n)$ if and only if $\exists k \in \llbracket 1, n \rrbracket$ and $\exists i \in \{-1,1\}$ such that $(f_1,\hdots,f_n)$ is equal to:
\begin{itemize}
\item either $\left(e_1, \hdots, e_k, e_k+1~(\textrm{mod}~4), \dots, e_n+1~(\textrm{mod}~4)\right)$
\item or $\left(e_1, \hdots, e_k, e_k-1~(\textrm{mod}~4), \dots, e_n-1~(\textrm{mod}~4)\right)$.
\end{itemize}
\end{itemize}

\begin{figure}
\centering
\scalebox{1} 
{
\begin{pspicture}(0,-3.3819609)(6.82,3.3819609)
\usefont{T1}{ptm}{m}{n}
\rput(0.9540625,2.343039){21}
\usefont{T1}{ptm}{m}{n}
\rput(2.5685937,2.343039){22}
\psbezier[linewidth=0.04,arrowsize=0.05291667cm 2.0,arrowlength=1.4,arrowinset=0.4]{->}(1.2,2.438039)(1.8,2.638039)(1.8,2.638039)(2.4,2.438039)
\usefont{T1}{ptm}{m}{n}
\rput(4.1610937,2.343039){23}
\psbezier[linewidth=0.04,arrowsize=0.05291667cm 2.0,arrowlength=1.4,arrowinset=0.4]{->}(2.8,2.438039)(3.4,2.638039)(3.4,2.638039)(4.0,2.438039)
\usefont{T1}{ptm}{m}{n}
\rput(5.7696877,2.343039){20}
\psbezier[linewidth=0.04,arrowsize=0.05291667cm 2.0,arrowlength=1.4,arrowinset=0.4]{->}(4.4,2.438039)(5.0,2.638039)(5.0,2.638039)(5.6,2.438039)
\psbezier[linewidth=0.04,linestyle=dotted,dotsep=0.16cm,arrowsize=0.05291667cm 2.0,arrowlength=1.4,arrowinset=0.4]{->}(0.0,2.438039)(0.6,2.638039)(0.2,2.638039)(0.8,2.438039)
\psbezier[linewidth=0.04,arrowsize=0.05291667cm 2.0,arrowlength=1.4,arrowinset=0.4]{->}(2.4,2.238039)(1.8,2.038039)(1.8,2.038039)(1.2,2.238039)
\psbezier[linewidth=0.04,arrowsize=0.05291667cm 2.0,arrowlength=1.4,arrowinset=0.4]{->}(4.0,2.238039)(3.4,2.038039)(3.4,2.038039)(2.8,2.238039)
\psbezier[linewidth=0.04,arrowsize=0.05291667cm 2.0,arrowlength=1.4,arrowinset=0.4]{->}(5.6,2.238039)(5.0,2.038039)(5.0,2.038039)(4.4,2.238039)
\psbezier[linewidth=0.04,linestyle=dotted,dotsep=0.16cm,arrowsize=0.05291667cm 2.0,arrowlength=1.4,arrowinset=0.4]{->}(0.8,2.238039)(0.4,2.238039)(0.4,2.038039)(0.0,2.238039)
\psbezier[linewidth=0.04,linestyle=dotted,dotsep=0.16cm,arrowsize=0.05291667cm 2.0,arrowlength=1.4,arrowinset=0.4]{->}(6.0,2.438039)(6.6,2.638039)(6.2,2.638039)(6.8,2.438039)
\psbezier[linewidth=0.04,linestyle=dotted,dotsep=0.16cm,arrowsize=0.05291667cm 2.0,arrowlength=1.4,arrowinset=0.4]{->}(6.8,2.238039)(6.4,2.238039)(6.4,2.038039)(6.0,2.238039)
\usefont{T1}{ptm}{m}{n}
\rput(0.9525,0.7430391){10}
\usefont{T1}{ptm}{m}{n}
\rput(2.536875,0.7430391){11}
\psbezier[linewidth=0.04,arrowsize=0.05291667cm 2.0,arrowlength=1.4,arrowinset=0.4]{->}(1.2,0.83803904)(1.8,1.0380391)(1.8,1.0380391)(2.4,0.83803904)
\usefont{T1}{ptm}{m}{n}
\rput(4.1514063,0.7430391){12}
\psbezier[linewidth=0.04,arrowsize=0.05291667cm 2.0,arrowlength=1.4,arrowinset=0.4]{->}(2.8,0.83803904)(3.4,1.0380391)(3.4,1.0380391)(4.0,0.83803904)
\usefont{T1}{ptm}{m}{n}
\rput(5.743906,0.7430391){13}
\psbezier[linewidth=0.04,arrowsize=0.05291667cm 2.0,arrowlength=1.4,arrowinset=0.4]{->}(4.4,0.83803904)(5.0,1.0380391)(5.0,1.0380391)(5.6,0.83803904)
\psbezier[linewidth=0.04,linestyle=dotted,dotsep=0.16cm,arrowsize=0.05291667cm 2.0,arrowlength=1.4,arrowinset=0.4]{->}(0.0,0.83803904)(0.6,1.0380391)(0.2,1.0380391)(0.8,0.83803904)
\psbezier[linewidth=0.04,arrowsize=0.05291667cm 2.0,arrowlength=1.4,arrowinset=0.4]{->}(2.4,0.63803905)(1.8,0.43803906)(1.8,0.43803906)(1.2,0.63803905)
\psbezier[linewidth=0.04,arrowsize=0.05291667cm 2.0,arrowlength=1.4,arrowinset=0.4]{->}(4.0,0.63803905)(3.4,0.43803906)(3.4,0.43803906)(2.8,0.63803905)
\psbezier[linewidth=0.04,arrowsize=0.05291667cm 2.0,arrowlength=1.4,arrowinset=0.4]{->}(5.6,0.63803905)(5.0,0.43803906)(5.0,0.43803906)(4.4,0.63803905)
\psbezier[linewidth=0.04,linestyle=dotted,dotsep=0.16cm,arrowsize=0.05291667cm 2.0,arrowlength=1.4,arrowinset=0.4]{->}(0.8,0.63803905)(0.4,0.63803905)(0.4,0.43803906)(0.0,0.63803905)
\psbezier[linewidth=0.04,linestyle=dotted,dotsep=0.16cm,arrowsize=0.05291667cm 2.0,arrowlength=1.4,arrowinset=0.4]{->}(6.0,0.83803904)(6.6,1.0380391)(6.2,1.0380391)(6.8,0.83803904)
\psbezier[linewidth=0.04,linestyle=dotted,dotsep=0.16cm,arrowsize=0.05291667cm 2.0,arrowlength=1.4,arrowinset=0.4]{->}(6.8,0.63803905)(6.4,0.63803905)(6.4,0.43803906)(6.0,0.63803905)
\psbezier[linewidth=0.04,arrowsize=0.05291667cm 2.0,arrowlength=1.4,arrowinset=0.4]{->}(0.8128585,0.9631987)(0.64030945,1.5716614)(0.64030945,1.5716614)(0.8673477,2.1619608)
\psbezier[linewidth=0.04,arrowsize=0.05291667cm 2.0,arrowlength=1.4,arrowinset=0.4]{->}(1.0671414,2.1528795)(1.2396905,1.5444168)(1.2396905,1.5444168)(1.0126523,0.9541172)
\psbezier[linewidth=0.04,arrowsize=0.05291667cm 2.0,arrowlength=1.4,arrowinset=0.4]{->}(2.4128585,0.9631987)(2.2403095,1.5716614)(2.2403095,1.5716614)(2.4673479,2.1619608)
\psbezier[linewidth=0.04,arrowsize=0.05291667cm 2.0,arrowlength=1.4,arrowinset=0.4]{->}(2.6671414,2.1528795)(2.8396904,1.5444168)(2.8396904,1.5444168)(2.6126523,0.9541172)
\psbezier[linewidth=0.04,arrowsize=0.05291667cm 2.0,arrowlength=1.4,arrowinset=0.4]{->}(4.0128584,0.9631987)(3.8403094,1.5716614)(3.8403094,1.5716614)(4.0673475,2.1619608)
\psbezier[linewidth=0.04,arrowsize=0.05291667cm 2.0,arrowlength=1.4,arrowinset=0.4]{->}(4.2671413,2.1528795)(4.4396906,1.5444168)(4.4396906,1.5444168)(4.212652,0.9541172)
\psbezier[linewidth=0.04,arrowsize=0.05291667cm 2.0,arrowlength=1.4,arrowinset=0.4]{->}(5.612859,0.9631987)(5.4403095,1.5716614)(5.4403095,1.5716614)(5.667348,2.1619608)
\psbezier[linewidth=0.04,arrowsize=0.05291667cm 2.0,arrowlength=1.4,arrowinset=0.4]{->}(5.8671412,2.1528795)(6.0396905,1.5444168)(6.0396905,1.5444168)(5.812652,0.9541172)
\psbezier[linewidth=0.04,arrowsize=0.05291667cm 2.0,arrowlength=1.4,arrowinset=0.4]{->}(0.8128585,-0.6368013)(0.64030945,-0.028338647)(0.64030945,-0.028338647)(0.8673477,0.56196094)
\psbezier[linewidth=0.04,arrowsize=0.05291667cm 2.0,arrowlength=1.4,arrowinset=0.4]{->}(1.0671414,0.5528794)(1.2396905,-0.055583242)(1.2396905,-0.055583242)(1.0126523,-0.64588284)
\psbezier[linewidth=0.04,arrowsize=0.05291667cm 2.0,arrowlength=1.4,arrowinset=0.4]{->}(2.4128585,-0.6368013)(2.2403095,-0.028338647)(2.2403095,-0.028338647)(2.4673479,0.56196094)
\psbezier[linewidth=0.04,arrowsize=0.05291667cm 2.0,arrowlength=1.4,arrowinset=0.4]{->}(2.6671414,0.5528794)(2.8396904,-0.055583242)(2.8396904,-0.055583242)(2.6126523,-0.64588284)
\psbezier[linewidth=0.04,arrowsize=0.05291667cm 2.0,arrowlength=1.4,arrowinset=0.4]{->}(4.0128584,-0.6368013)(3.8403094,-0.028338647)(3.8403094,-0.028338647)(4.0673475,0.56196094)
\psbezier[linewidth=0.04,arrowsize=0.05291667cm 2.0,arrowlength=1.4,arrowinset=0.4]{->}(4.2671413,0.5528794)(4.4396906,-0.055583242)(4.4396906,-0.055583242)(4.212652,-0.64588284)
\psbezier[linewidth=0.04,arrowsize=0.05291667cm 2.0,arrowlength=1.4,arrowinset=0.4]{->}(5.612859,-0.6368013)(5.4403095,-0.028338647)(5.4403095,-0.028338647)(5.667348,0.56196094)
\psbezier[linewidth=0.04,arrowsize=0.05291667cm 2.0,arrowlength=1.4,arrowinset=0.4]{->}(5.8671412,0.5528794)(6.0396905,-0.055583242)(6.0396905,-0.055583242)(5.812652,-0.64588284)
\usefont{T1}{ptm}{m}{n}
\rput(0.9584375,-0.85696095){03}
\usefont{T1}{ptm}{m}{n}
\rput(2.5670311,-0.85696095){00}
\psbezier[linewidth=0.04,arrowsize=0.05291667cm 2.0,arrowlength=1.4,arrowinset=0.4]{->}(1.2,-0.7619609)(1.8,-0.56196094)(1.8,-0.56196094)(2.4,-0.7619609)
\usefont{T1}{ptm}{m}{n}
\rput(4.1514063,-0.85696095){01}
\psbezier[linewidth=0.04,arrowsize=0.05291667cm 2.0,arrowlength=1.4,arrowinset=0.4]{->}(2.8,-0.7619609)(3.4,-0.56196094)(3.4,-0.56196094)(4.0,-0.7619609)
\usefont{T1}{ptm}{m}{n}
\rput(5.7659373,-0.85696095){02}
\psbezier[linewidth=0.04,arrowsize=0.05291667cm 2.0,arrowlength=1.4,arrowinset=0.4]{->}(4.4,-0.7619609)(5.0,-0.56196094)(5.0,-0.56196094)(5.6,-0.7619609)
\psbezier[linewidth=0.04,linestyle=dotted,dotsep=0.16cm,arrowsize=0.05291667cm 2.0,arrowlength=1.4,arrowinset=0.4]{->}(0.0,-0.7619609)(0.6,-0.56196094)(0.2,-0.56196094)(0.8,-0.7619609)
\psbezier[linewidth=0.04,arrowsize=0.05291667cm 2.0,arrowlength=1.4,arrowinset=0.4]{->}(2.4,-0.961961)(1.8,-1.161961)(1.8,-1.161961)(1.2,-0.961961)
\psbezier[linewidth=0.04,arrowsize=0.05291667cm 2.0,arrowlength=1.4,arrowinset=0.4]{->}(4.0,-0.961961)(3.4,-1.161961)(3.4,-1.161961)(2.8,-0.961961)
\psbezier[linewidth=0.04,arrowsize=0.05291667cm 2.0,arrowlength=1.4,arrowinset=0.4]{->}(5.6,-0.961961)(5.0,-1.161961)(5.0,-1.161961)(4.4,-0.961961)
\psbezier[linewidth=0.04,linestyle=dotted,dotsep=0.16cm,arrowsize=0.05291667cm 2.0,arrowlength=1.4,arrowinset=0.4]{->}(0.8,-0.961961)(0.4,-0.961961)(0.4,-1.161961)(0.0,-0.961961)
\psbezier[linewidth=0.04,linestyle=dotted,dotsep=0.16cm,arrowsize=0.05291667cm 2.0,arrowlength=1.4,arrowinset=0.4]{->}(6.0,-0.7619609)(6.6,-0.56196094)(6.2,-0.56196094)(6.8,-0.7619609)
\psbezier[linewidth=0.04,linestyle=dotted,dotsep=0.16cm,arrowsize=0.05291667cm 2.0,arrowlength=1.4,arrowinset=0.4]{->}(6.8,-0.961961)(6.4,-0.961961)(6.4,-1.161961)(6.0,-0.961961)
\usefont{T1}{ptm}{m}{n}
\rput(0.96515626,-2.456961){32}
\usefont{T1}{ptm}{m}{n}
\rput(2.5576563,-2.456961){33}
\psbezier[linewidth=0.04,arrowsize=0.05291667cm 2.0,arrowlength=1.4,arrowinset=0.4]{->}(1.2,-2.361961)(1.8,-2.1619608)(1.8,-2.1619608)(2.4,-2.361961)
\usefont{T1}{ptm}{m}{n}
\rput(4.16625,-2.456961){30}
\psbezier[linewidth=0.04,arrowsize=0.05291667cm 2.0,arrowlength=1.4,arrowinset=0.4]{->}(2.8,-2.361961)(3.4,-2.1619608)(3.4,-2.1619608)(4.0,-2.361961)
\usefont{T1}{ptm}{m}{n}
\rput(5.750625,-2.456961){31}
\psbezier[linewidth=0.04,arrowsize=0.05291667cm 2.0,arrowlength=1.4,arrowinset=0.4]{->}(4.4,-2.361961)(5.0,-2.1619608)(5.0,-2.1619608)(5.6,-2.361961)
\psbezier[linewidth=0.04,linestyle=dotted,dotsep=0.16cm,arrowsize=0.05291667cm 2.0,arrowlength=1.4,arrowinset=0.4]{->}(0.0,-2.361961)(0.6,-2.1619608)(0.2,-2.1619608)(0.8,-2.361961)
\psbezier[linewidth=0.04,arrowsize=0.05291667cm 2.0,arrowlength=1.4,arrowinset=0.4]{->}(2.4,-2.561961)(1.8,-2.761961)(1.8,-2.761961)(1.2,-2.561961)
\psbezier[linewidth=0.04,arrowsize=0.05291667cm 2.0,arrowlength=1.4,arrowinset=0.4]{->}(4.0,-2.561961)(3.4,-2.761961)(3.4,-2.761961)(2.8,-2.561961)
\psbezier[linewidth=0.04,arrowsize=0.05291667cm 2.0,arrowlength=1.4,arrowinset=0.4]{->}(5.6,-2.561961)(5.0,-2.761961)(5.0,-2.761961)(4.4,-2.561961)
\psbezier[linewidth=0.04,linestyle=dotted,dotsep=0.16cm,arrowsize=0.05291667cm 2.0,arrowlength=1.4,arrowinset=0.4]{->}(0.8,-2.561961)(0.4,-2.561961)(0.4,-2.761961)(0.0,-2.561961)
\psbezier[linewidth=0.04,linestyle=dotted,dotsep=0.16cm,arrowsize=0.05291667cm 2.0,arrowlength=1.4,arrowinset=0.4]{->}(6.0,-2.361961)(6.6,-2.1619608)(6.2,-2.1619608)(6.8,-2.361961)
\psbezier[linewidth=0.04,linestyle=dotted,dotsep=0.16cm,arrowsize=0.05291667cm 2.0,arrowlength=1.4,arrowinset=0.4]{->}(6.8,-2.561961)(6.4,-2.561961)(6.4,-2.761961)(6.0,-2.561961)
\psbezier[linewidth=0.04,arrowsize=0.05291667cm 2.0,arrowlength=1.4,arrowinset=0.4]{->}(0.8128585,-2.2368014)(0.64030945,-1.6283387)(0.64030945,-1.6283387)(0.8673477,-1.0380391)
\psbezier[linewidth=0.04,arrowsize=0.05291667cm 2.0,arrowlength=1.4,arrowinset=0.4]{->}(1.0671414,-1.0471206)(1.2396905,-1.6555833)(1.2396905,-1.6555833)(1.0126523,-2.2458827)
\psbezier[linewidth=0.04,arrowsize=0.05291667cm 2.0,arrowlength=1.4,arrowinset=0.4]{->}(2.4128585,-2.2368014)(2.2403095,-1.6283387)(2.2403095,-1.6283387)(2.4673479,-1.0380391)
\psbezier[linewidth=0.04,arrowsize=0.05291667cm 2.0,arrowlength=1.4,arrowinset=0.4]{->}(2.6671414,-1.0471206)(2.8396904,-1.6555833)(2.8396904,-1.6555833)(2.6126523,-2.2458827)
\psbezier[linewidth=0.04,arrowsize=0.05291667cm 2.0,arrowlength=1.4,arrowinset=0.4]{->}(4.0128584,-2.2368014)(3.8403094,-1.6283387)(3.8403094,-1.6283387)(4.0673475,-1.0380391)
\psbezier[linewidth=0.04,arrowsize=0.05291667cm 2.0,arrowlength=1.4,arrowinset=0.4]{->}(4.2671413,-1.0471206)(4.4396906,-1.6555833)(4.4396906,-1.6555833)(4.212652,-2.2458827)
\psbezier[linewidth=0.04,arrowsize=0.05291667cm 2.0,arrowlength=1.4,arrowinset=0.4]{->}(5.612859,-2.2368014)(5.4403095,-1.6283387)(5.4403095,-1.6283387)(5.667348,-1.0380391)
\psbezier[linewidth=0.04,arrowsize=0.05291667cm 2.0,arrowlength=1.4,arrowinset=0.4]{->}(5.8671412,-1.0471206)(6.0396905,-1.6555833)(6.0396905,-1.6555833)(5.812652,-2.2458827)
\psbezier[linewidth=0.04,linestyle=dotted,dotsep=0.16cm,arrowsize=0.05291667cm 2.0,arrowlength=1.4,arrowinset=0.4]{->}(0.8,-3.361961)(0.64030945,-2.8283386)(0.64030945,-3.2283387)(0.8673477,-2.638039)
\psbezier[linewidth=0.04,linestyle=dotted,dotsep=0.16cm,arrowsize=0.05291667cm 2.0,arrowlength=1.4,arrowinset=0.4]{->}(1.0671414,-2.6471205)(1.2396905,-3.0099561)(1.2396905,-3.0099561)(1.0126523,-3.361961)
\psbezier[linewidth=0.04,linestyle=dotted,dotsep=0.16cm,arrowsize=0.05291667cm 2.0,arrowlength=1.4,arrowinset=0.4]{->}(2.4,-3.361961)(2.2403095,-2.8283386)(2.2403095,-3.2283387)(2.4673479,-2.638039)
\psbezier[linewidth=0.04,linestyle=dotted,dotsep=0.16cm,arrowsize=0.05291667cm 2.0,arrowlength=1.4,arrowinset=0.4]{->}(2.6671414,-2.6471205)(2.8396904,-3.0099561)(2.8396904,-3.0099561)(2.6126523,-3.361961)
\psbezier[linewidth=0.04,linestyle=dotted,dotsep=0.16cm,arrowsize=0.05291667cm 2.0,arrowlength=1.4,arrowinset=0.4]{->}(4.0,-3.361961)(3.8403094,-2.8283386)(3.8403094,-3.2283387)(4.0673475,-2.638039)
\psbezier[linewidth=0.04,linestyle=dotted,dotsep=0.16cm,arrowsize=0.05291667cm 2.0,arrowlength=1.4,arrowinset=0.4]{->}(4.2671413,-2.6471205)(4.4396906,-3.0099561)(4.4396906,-3.0099561)(4.212652,-3.361961)
\psbezier[linewidth=0.04,linestyle=dotted,dotsep=0.16cm,arrowsize=0.05291667cm 2.0,arrowlength=1.4,arrowinset=0.4]{->}(5.6,-3.361961)(5.4403095,-2.8283386)(5.4403095,-3.2283387)(5.667348,-2.638039)
\psbezier[linewidth=0.04,linestyle=dotted,dotsep=0.16cm,arrowsize=0.05291667cm 2.0,arrowlength=1.4,arrowinset=0.4]{->}(5.8671412,-2.6471205)(6.0396905,-3.0099561)(6.0396905,-3.0099561)(5.812652,-3.361961)
\psbezier[linewidth=0.04,linestyle=dotted,dotsep=0.16cm,arrowsize=0.05291667cm 2.0,arrowlength=1.4,arrowinset=0.4]{->}(0.8,2.638039)(0.64030945,3.1716614)(0.64030945,2.7716613)(0.8673477,3.361961)
\psbezier[linewidth=0.04,linestyle=dotted,dotsep=0.16cm,arrowsize=0.05291667cm 2.0,arrowlength=1.4,arrowinset=0.4]{->}(1.0671414,3.3528795)(1.2396905,2.9900439)(1.2396905,2.9900439)(1.0126523,2.638039)
\psbezier[linewidth=0.04,linestyle=dotted,dotsep=0.16cm,arrowsize=0.05291667cm 2.0,arrowlength=1.4,arrowinset=0.4]{->}(2.4,2.638039)(2.2403095,3.1716614)(2.2403095,2.7716613)(2.4673479,3.361961)
\psbezier[linewidth=0.04,linestyle=dotted,dotsep=0.16cm,arrowsize=0.05291667cm 2.0,arrowlength=1.4,arrowinset=0.4]{->}(2.6671414,3.3528795)(2.8396904,2.9900439)(2.8396904,2.9900439)(2.6126523,2.638039)
\psbezier[linewidth=0.04,linestyle=dotted,dotsep=0.16cm,arrowsize=0.05291667cm 2.0,arrowlength=1.4,arrowinset=0.4]{->}(4.0,2.638039)(3.8403094,3.1716614)(3.8403094,2.7716613)(4.0673475,3.361961)
\psbezier[linewidth=0.04,linestyle=dotted,dotsep=0.16cm,arrowsize=0.05291667cm 2.0,arrowlength=1.4,arrowinset=0.4]{->}(4.2671413,3.3528795)(4.4396906,2.9900439)(4.4396906,2.9900439)(4.212652,2.638039)
\psbezier[linewidth=0.04,linestyle=dotted,dotsep=0.16cm,arrowsize=0.05291667cm 2.0,arrowlength=1.4,arrowinset=0.4]{->}(5.6,2.638039)(5.4403095,3.1716614)(5.4403095,2.7716613)(5.667348,3.361961)
\psbezier[linewidth=0.04,linestyle=dotted,dotsep=0.16cm,arrowsize=0.05291667cm 2.0,arrowlength=1.4,arrowinset=0.4]{->}(5.8671412,3.3528795)(6.0396905,2.9900439)(6.0396905,2.9900439)(5.812652,2.638039)
\end{pspicture} 
}
\caption{The digraph $\mathfrak{G}_0(2)$}
\label{Go2}
\end{figure}
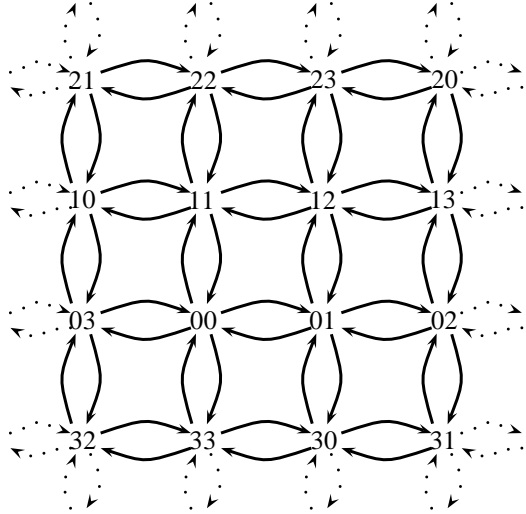

Obviously, in $\mathfrak{G}_0(n)$, if there is a directed edge from the vertex $i$ to the vertex $j$, then there is another edge from $j$ to $i$ too.
Such a graph is depicted in Fig.~\ref{Go2}, in which some edges are dotted to represent 
the fact that this graph is as a torus: we can go from the vertex 22 to the vertex 33 for
instance.
The rule of construction of this graph is detailed in Figure~\ref{Go2Rules}. 

\begin{figure}
\centering
\scalebox{0.6} 
{
\begin{pspicture}(0,-2.313125)(10.570625,2.313125)
\usefont{T1}{ptm}{m}{n}
\rput(5.180625,0.03125){\LARGE (i,j)}
\psline[linewidth=0.04cm,arrowsize=0.05291667cm 2.0,arrowlength=1.4,arrowinset=0.4]{->}(5.1153126,0.35625)(5.1153126,1.55625)
\psline[linewidth=0.04cm,arrowsize=0.05291667cm 2.0,arrowlength=1.4,arrowinset=0.4]{->}(5.7153125,-0.04375)(6.9153123,-0.04375)
\psline[linewidth=0.04cm,arrowsize=0.05291667cm 2.0,arrowlength=1.4,arrowinset=0.4]{->}(5.1153126,-0.44375)(5.1153126,-1.64375)
\psline[linewidth=0.04cm,arrowsize=0.05291667cm 2.0,arrowlength=1.4,arrowinset=0.4]{->}(4.7153125,-0.04375)(3.5153124,-0.04375)
\usefont{T1}{ptm}{m}{n}
\rput(5.180625,2.03125){\LARGE ((i+1) mod 4, (j+1) mod 4)}
\usefont{T1}{ptm}{m}{n}
\rput(8.740625,0.03125){\LARGE (i,(j+1) mod 4)}
\usefont{T1}{ptm}{m}{n}
\rput(1.670625,0.03125){\LARGE (i,(j-1) mod 4)}
\usefont{T1}{ptm}{m}{n}
\rput(5.040625,-1.96875){\LARGE ((i-1) mod 4, (j-1) mod 4)}
\end{pspicture} 
}
\caption{Rules of $\mathfrak{G}_0(2)$}
\label{Go2Rules}
\end{figure}
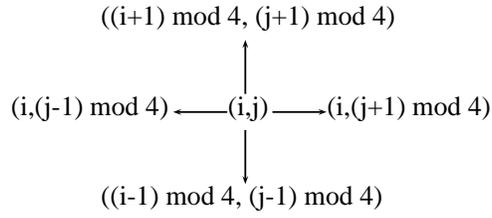

\begin{figure}
\centering
\scalebox{0.5} 
{
\begin{pspicture}(0,-2.613125)(21.850624,2.613125)
\usefont{T1}{ptm}{m}{n}
\rput(11.600625,0.13125){\LARGE (i,j,k)}
\psline[linewidth=0.04cm,arrowsize=0.05291667cm 2.0,arrowlength=1.4,arrowinset=0.4]{->}(11.515312,0.45625)(11.515312,1.85625)
\psline[linewidth=0.04cm,arrowsize=0.05291667cm 2.0,arrowlength=1.4,arrowinset=0.4]{->}(12.515312,0.05625)(13.715313,0.05625)
\psline[linewidth=0.04cm,arrowsize=0.05291667cm 2.0,arrowlength=1.4,arrowinset=0.4]{->}(11.515312,-0.34375)(11.515312,-1.94375)
\psline[linewidth=0.04cm,arrowsize=0.05291667cm 2.0,arrowlength=1.4,arrowinset=0.4]{->}(10.915313,0.05625)(9.715313,0.05625)
\usefont{T1}{ptm}{m}{n}
\rput(11.240625,2.33125){\LARGE ((i+1) mod 4, (j+1) mod 4, (k+1) mod 4)}
\usefont{T1}{ptm}{m}{n}
\rput(15.760625,0.13125){\LARGE (i,j,(k+1) mod 4)}
\usefont{T1}{ptm}{m}{n}
\rput(7.690625,0.13125){\LARGE (i,j,(k-1) mod 4)}
\usefont{T1}{ptm}{m}{n}
\rput(11.430625,-2.26875){\LARGE ((i-1) mod 4, (j-1) mod 4, (k-1) mod 4)}
\psline[linewidth=0.04cm,arrowsize=0.05291667cm 2.0,arrowlength=1.4,arrowinset=0.4]{->}(12.515312,0.25625)(16.915312,1.25625)
\psline[linewidth=0.04cm,arrowsize=0.05291667cm 2.0,arrowlength=1.4,arrowinset=0.4]{->}(10.915313,-0.14375)(6.3153124,-1.14375)
\usefont{T1}{ptm}{m}{n}
\rput(18.280624,1.53125){\LARGE (i, (j+1) mod 4, (k+1) mod 4)}
\usefont{T1}{ptm}{m}{n}
\rput(3.340625,-1.46875){\LARGE (i, (j-1) mod 4, (k-1) mod 4)}
\end{pspicture} 
}
\caption{Rules of $\mathfrak{G}_0(3)$}
\label{Go3Rules}
\end{figure}
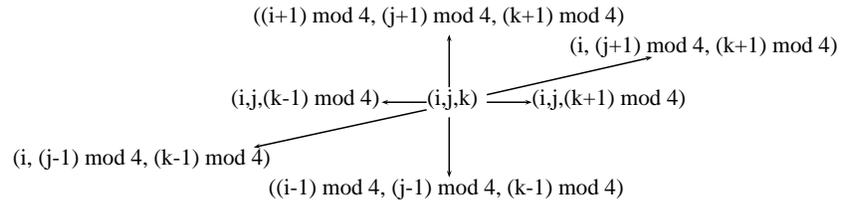

\begin{figure}
\centering
\scalebox{1} 
{
\begin{pspicture}(0,-3.3819609)(3.819381,3.3819609)
\usefont{T1}{ptm}{m}{n}
\rput(0.31375307,2.343039){21}
\usefont{T1}{ptm}{m}{n}
\rput(1.9282843,2.343039){22}
\psbezier[linewidth=0.04,arrowsize=0.05291667cm 2.0,arrowlength=1.4,arrowinset=0.4]{->}(0.55969054,2.438039)(1.1596906,2.638039)(1.1596906,2.638039)(1.7596905,2.438039)
\usefont{T1}{ptm}{m}{n}
\rput(3.5207844,2.343039){23}
\psbezier[linewidth=0.04,arrowsize=0.05291667cm 2.0,arrowlength=1.4,arrowinset=0.4]{->}(2.1596906,2.438039)(2.7596905,2.638039)(2.7596905,2.638039)(3.3596907,2.438039)
\psbezier[linewidth=0.04,arrowsize=0.05291667cm 2.0,arrowlength=1.4,arrowinset=0.4]{->}(1.7596905,2.238039)(1.1596906,2.038039)(1.1596906,2.038039)(0.55969054,2.238039)
\psbezier[linewidth=0.04,arrowsize=0.05291667cm 2.0,arrowlength=1.4,arrowinset=0.4]{->}(3.3596907,2.238039)(2.7596905,2.038039)(2.7596905,2.038039)(2.1596906,2.238039)
\usefont{T1}{ptm}{m}{n}
\rput(0.31219056,0.7430391){10}
\usefont{T1}{ptm}{m}{n}
\rput(1.8965656,0.7430391){11}
\psbezier[linewidth=0.04,arrowsize=0.05291667cm 2.0,arrowlength=1.4,arrowinset=0.4]{->}(0.55969054,0.83803904)(1.1596906,1.0380391)(1.1596906,1.0380391)(1.7596905,0.83803904)
\usefont{T1}{ptm}{m}{n}
\rput(3.5110967,0.7430391){12}
\psbezier[linewidth=0.04,arrowsize=0.05291667cm 2.0,arrowlength=1.4,arrowinset=0.4]{->}(2.1596906,0.83803904)(2.7596905,1.0380391)(2.7596905,1.0380391)(3.3596907,0.83803904)
\psbezier[linewidth=0.04,arrowsize=0.05291667cm 2.0,arrowlength=1.4,arrowinset=0.4]{->}(1.7596905,0.63803905)(1.1596906,0.43803906)(1.1596906,0.43803906)(0.55969054,0.63803905)
\psbezier[linewidth=0.04,arrowsize=0.05291667cm 2.0,arrowlength=1.4,arrowinset=0.4]{->}(3.3596907,0.63803905)(2.7596905,0.43803906)(2.7596905,0.43803906)(2.1596906,0.63803905)
\psbezier[linewidth=0.04,arrowsize=0.05291667cm 2.0,arrowlength=1.4,arrowinset=0.4]{->}(0.17254911,0.9631987)(0.0,1.5716614)(0.0,1.5716614)(0.22703831,2.1619608)
\psbezier[linewidth=0.04,arrowsize=0.05291667cm 2.0,arrowlength=1.4,arrowinset=0.4]{->}(0.42683202,2.1528795)(0.59938115,1.5444168)(0.59938115,1.5444168)(0.37234282,0.9541172)
\psbezier[linewidth=0.04,arrowsize=0.05291667cm 2.0,arrowlength=1.4,arrowinset=0.4]{->}(1.7725492,0.9631987)(1.6,1.5716614)(1.6,1.5716614)(1.8270383,2.1619608)
\psbezier[linewidth=0.04,arrowsize=0.05291667cm 2.0,arrowlength=1.4,arrowinset=0.4]{->}(2.026832,2.1528795)(2.199381,1.5444168)(2.199381,1.5444168)(1.9723428,0.9541172)
\psbezier[linewidth=0.04,arrowsize=0.05291667cm 2.0,arrowlength=1.4,arrowinset=0.4]{->}(3.372549,0.9631987)(3.2,1.5716614)(3.2,1.5716614)(3.4270382,2.1619608)
\psbezier[linewidth=0.04,arrowsize=0.05291667cm 2.0,arrowlength=1.4,arrowinset=0.4]{->}(3.626832,2.1528795)(3.799381,1.5444168)(3.799381,1.5444168)(3.5723429,0.9541172)
\psbezier[linewidth=0.04,arrowsize=0.05291667cm 2.0,arrowlength=1.4,arrowinset=0.4]{->}(0.17254911,-0.6368013)(0.0,-0.028338647)(0.0,-0.028338647)(0.22703831,0.56196094)
\psbezier[linewidth=0.04,arrowsize=0.05291667cm 2.0,arrowlength=1.4,arrowinset=0.4]{->}(0.42683202,0.5528794)(0.59938115,-0.055583242)(0.59938115,-0.055583242)(0.37234282,-0.64588284)
\psbezier[linewidth=0.04,arrowsize=0.05291667cm 2.0,arrowlength=1.4,arrowinset=0.4]{->}(1.7725492,-0.6368013)(1.6,-0.028338647)(1.6,-0.028338647)(1.8270383,0.56196094)
\psbezier[linewidth=0.04,arrowsize=0.05291667cm 2.0,arrowlength=1.4,arrowinset=0.4]{->}(2.026832,0.5528794)(2.199381,-0.055583242)(2.199381,-0.055583242)(1.9723428,-0.64588284)
\psbezier[linewidth=0.04,arrowsize=0.05291667cm 2.0,arrowlength=1.4,arrowinset=0.4]{->}(3.372549,-0.6368013)(3.2,-0.028338647)(3.2,-0.028338647)(3.4270382,0.56196094)
\psbezier[linewidth=0.04,arrowsize=0.05291667cm 2.0,arrowlength=1.4,arrowinset=0.4]{->}(3.626832,0.5528794)(3.799381,-0.055583242)(3.799381,-0.055583242)(3.5723429,-0.64588284)
\usefont{T1}{ptm}{m}{n}
\rput(0.31812805,-0.85696095){03}
\usefont{T1}{ptm}{m}{n}
\rput(1.9267218,-0.85696095){00}
\psbezier[linewidth=0.04,arrowsize=0.05291667cm 2.0,arrowlength=1.4,arrowinset=0.4]{->}(0.55969054,-0.7619609)(1.1596906,-0.56196094)(1.1596906,-0.56196094)(1.7596905,-0.7619609)
\usefont{T1}{ptm}{m}{n}
\rput(3.5110967,-0.85696095){01}
\psbezier[linewidth=0.04,arrowsize=0.05291667cm 2.0,arrowlength=1.4,arrowinset=0.4]{->}(2.1596906,-0.7619609)(2.7596905,-0.56196094)(2.7596905,-0.56196094)(3.3596907,-0.7619609)
\psbezier[linewidth=0.04,arrowsize=0.05291667cm 2.0,arrowlength=1.4,arrowinset=0.4]{->}(1.7596905,-0.961961)(1.1596906,-1.161961)(1.1596906,-1.161961)(0.55969054,-0.961961)
\psbezier[linewidth=0.04,arrowsize=0.05291667cm 2.0,arrowlength=1.4,arrowinset=0.4]{->}(3.3596907,-0.961961)(2.7596905,-1.161961)(2.7596905,-1.161961)(2.1596906,-0.961961)
\usefont{T1}{ptm}{m}{n}
\rput(0.3248468,-2.456961){32}
\usefont{T1}{ptm}{m}{n}
\rput(1.9173468,-2.456961){33}
\psbezier[linewidth=0.04,arrowsize=0.05291667cm 2.0,arrowlength=1.4,arrowinset=0.4]{->}(0.55969054,-2.361961)(1.1596906,-2.1619608)(1.1596906,-2.1619608)(1.7596905,-2.361961)
\usefont{T1}{ptm}{m}{n}
\rput(3.5259407,-2.456961){30}
\psbezier[linewidth=0.04,arrowsize=0.05291667cm 2.0,arrowlength=1.4,arrowinset=0.4]{->}(2.1596906,-2.361961)(2.7596905,-2.1619608)(2.7596905,-2.1619608)(3.3596907,-2.361961)
\psbezier[linewidth=0.04,arrowsize=0.05291667cm 2.0,arrowlength=1.4,arrowinset=0.4]{->}(1.7596905,-2.561961)(1.1596906,-2.761961)(1.1596906,-2.761961)(0.55969054,-2.561961)
\psbezier[linewidth=0.04,arrowsize=0.05291667cm 2.0,arrowlength=1.4,arrowinset=0.4]{->}(3.3596907,-2.561961)(2.7596905,-2.761961)(2.7596905,-2.761961)(2.1596906,-2.561961)
\psbezier[linewidth=0.04,arrowsize=0.05291667cm 2.0,arrowlength=1.4,arrowinset=0.4]{->}(0.17254911,-2.2368014)(0.0,-1.6283387)(0.0,-1.6283387)(0.22703831,-1.0380391)
\psbezier[linewidth=0.04,arrowsize=0.05291667cm 2.0,arrowlength=1.4,arrowinset=0.4]{->}(0.42683202,-1.0471206)(0.59938115,-1.6555833)(0.59938115,-1.6555833)(0.37234282,-2.2458827)
\psbezier[linewidth=0.04,arrowsize=0.05291667cm 2.0,arrowlength=1.4,arrowinset=0.4]{->}(1.7725492,-2.2368014)(1.6,-1.6283387)(1.6,-1.6283387)(1.8270383,-1.0380391)
\psbezier[linewidth=0.04,arrowsize=0.05291667cm 2.0,arrowlength=1.4,arrowinset=0.4]{->}(2.026832,-1.0471206)(2.199381,-1.6555833)(2.199381,-1.6555833)(1.9723428,-2.2458827)
\psbezier[linewidth=0.04,arrowsize=0.05291667cm 2.0,arrowlength=1.4,arrowinset=0.4]{->}(3.372549,-2.2368014)(3.2,-1.6283387)(3.2,-1.6283387)(3.4270382,-1.0380391)
\psbezier[linewidth=0.04,arrowsize=0.05291667cm 2.0,arrowlength=1.4,arrowinset=0.4]{->}(3.626832,-1.0471206)(3.799381,-1.6555833)(3.799381,-1.6555833)(3.5723429,-2.2458827)
\psbezier[linewidth=0.04,linestyle=dotted,dotsep=0.16cm,arrowsize=0.05291667cm 2.0,arrowlength=1.4,arrowinset=0.4]{->}(0.15969056,-3.361961)(0.0,-2.8283386)(0.0,-3.2283387)(0.22703831,-2.638039)
\psbezier[linewidth=0.04,linestyle=dotted,dotsep=0.16cm,arrowsize=0.05291667cm 2.0,arrowlength=1.4,arrowinset=0.4]{->}(0.42683202,-2.6471205)(0.59938115,-3.0099561)(0.59938115,-3.0099561)(0.37234282,-3.361961)
\psbezier[linewidth=0.04,linestyle=dotted,dotsep=0.16cm,arrowsize=0.05291667cm 2.0,arrowlength=1.4,arrowinset=0.4]{->}(1.7596905,-3.361961)(1.6,-2.8283386)(1.6,-3.2283387)(1.8270383,-2.638039)
\psbezier[linewidth=0.04,linestyle=dotted,dotsep=0.16cm,arrowsize=0.05291667cm 2.0,arrowlength=1.4,arrowinset=0.4]{->}(2.026832,-2.6471205)(2.199381,-3.0099561)(2.199381,-3.0099561)(1.9723428,-3.361961)
\psbezier[linewidth=0.04,linestyle=dotted,dotsep=0.16cm,arrowsize=0.05291667cm 2.0,arrowlength=1.4,arrowinset=0.4]{->}(3.3596907,-3.361961)(3.2,-2.8283386)(3.2,-3.2283387)(3.4270382,-2.638039)
\psbezier[linewidth=0.04,linestyle=dotted,dotsep=0.16cm,arrowsize=0.05291667cm 2.0,arrowlength=1.4,arrowinset=0.4]{->}(3.626832,-2.6471205)(3.799381,-3.0099561)(3.799381,-3.0099561)(3.5723429,-3.361961)
\psbezier[linewidth=0.04,linestyle=dotted,dotsep=0.16cm,arrowsize=0.05291667cm 2.0,arrowlength=1.4,arrowinset=0.4]{->}(0.15969056,2.638039)(0.0,3.1716614)(0.0,2.7716613)(0.22703831,3.361961)
\psbezier[linewidth=0.04,linestyle=dotted,dotsep=0.16cm,arrowsize=0.05291667cm 2.0,arrowlength=1.4,arrowinset=0.4]{->}(0.42683202,3.3528795)(0.59938115,2.9900439)(0.59938115,2.9900439)(0.37234282,2.638039)
\psbezier[linewidth=0.04,linestyle=dotted,dotsep=0.16cm,arrowsize=0.05291667cm 2.0,arrowlength=1.4,arrowinset=0.4]{->}(1.7596905,2.638039)(1.6,3.1716614)(1.6,2.7716613)(1.8270383,3.361961)
\psbezier[linewidth=0.04,linestyle=dotted,dotsep=0.16cm,arrowsize=0.05291667cm 2.0,arrowlength=1.4,arrowinset=0.4]{->}(2.026832,3.3528795)(2.199381,2.9900439)(2.199381,2.9900439)(1.9723428,2.638039)
\psbezier[linewidth=0.04,linestyle=dotted,dotsep=0.16cm,arrowsize=0.05291667cm 2.0,arrowlength=1.4,arrowinset=0.4]{->}(3.3596907,2.638039)(3.2,3.1716614)(3.2,2.7716613)(3.4270382,3.361961)
\psbezier[linewidth=0.04,linestyle=dotted,dotsep=0.16cm,arrowsize=0.05291667cm 2.0,arrowlength=1.4,arrowinset=0.4]{->}(3.626832,3.3528795)(3.799381,2.9900439)(3.799381,2.9900439)(3.5723429,2.638039)
\end{pspicture} 
}
\caption{The digraph $\mathfrak{G}(2)$}
\label{G2}
\end{figure}
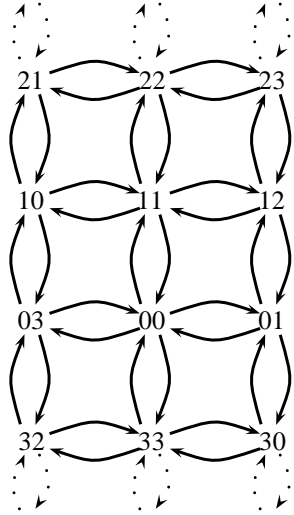

Let us now define another digraph as follows. $\mathfrak{G}(n)$ is the subgraph 
of $\mathfrak{G}_0(n)$ obtained by removing the vertices that do not correspond to 
a ``path without crossing'' according to Madras and Slade~\cite{Madras}. 
In other words, we remove from $\mathfrak{G}_0(n)$ vertices that
do not satisfy the $SAW_1(n)$ requirement. For instance, the digraph
 $\mathfrak{G}(2)$ associated to $\mathfrak{G}_0(2)$ is depicted in Fig.~\ref{G2},
whereas Figure~\ref{G3} contains both $\mathfrak{G}(3)$ and the removed vertices.
Its construction rules are recalled in Fig.~\ref{Go3Rules}.

The links between $\mathfrak{G}(n)$ and the SAW requirements can be summarized
as follows:
\begin{itemize}
\item The vertices of the graph $\mathfrak{G}_0(n)$ represent all the possible walks of 
length $n$ in the 2D square lattice.
\item The vertices that are preserved in $\mathfrak{G}(n)$ are the conformations of $SAW_1(n) = SAW_2(n)$.
\item Two adjacent vertices $i$ and $j$ in $\mathfrak{G}(n)$ are such that it is 
possible to change the conformation $i$ in $j$ in only one pivot move.
\item Finally, a conformation of $SAW_3(n)$ is a vertex of $\mathfrak{G}(n)$ that
is reachable from the vertex $000\hdots 0$ by following a path in $\mathfrak{G}(n)$.
\end{itemize}

For instance, the conformation $(2,2,3)$ is in $SAW_3(3)$ because we can find a walk
from $000$ to $223$ in $\mathfrak{G}(n)$. The following result is obvious,

\begin{theorem}
$SAW_3(n)$ corresponds to the connected component of $000\hdots 0$ in $\mathfrak{G}(n)$,
whereas $SAW_2(n)$ is the set of vertices of $\mathfrak{G}(n)$. Thus we have:

\begin{center}
$SAW_2(n) = SAW_3(n) \Longleftrightarrow \mathfrak{G}(n)$ is (strongly) connected.
\end{center}
\end{theorem}

The previous section shows that the connected component of $000\hdots 0$ in $\mathfrak{G}(158)$, $\mathfrak{G}(168)$, $\mathfrak{G}(175)$, and $\mathfrak{G}(914)$ are not
equal to $\mathfrak{G}(158)$, $\mathfrak{G}(168)$, $\mathfrak{G}(175)$, and $\mathfrak{G}(914)$ respectively.
In other words, these graphs are not connected ones. 

\begin{figure}
\centering
\scalebox{0.8} 
{
\begin{pspicture}(0,-7.1428127)(15.1325,7.1228123)
\usefont{T1}{ptm}{m}{n}
\rput(5.7246876,6.8478127){\Large 203}
\psline[linewidth=0.04cm,arrowsize=0.05291667cm 2.0,arrowlength=1.4,arrowinset=0.4]{<->}(6.1896877,6.9028125)(8.389688,6.9028125)
\usefont{T1}{ptm}{m}{n}
\rput(8.735937,6.8478127){\Large 200}
\psline[linewidth=0.04cm,arrowsize=0.05291667cm 2.0,arrowlength=1.4,arrowinset=0.4]{<->}(9.189688,6.9028125)(11.389688,6.9028125)
\usefont{T1}{ptm}{m}{n}
\rput(11.715313,6.8478127){\Large 201}
\psline[linewidth=0.04cm,arrowsize=0.05291667cm 2.0,arrowlength=1.4,arrowinset=0.4]{<->}(12.189688,6.9028125)(14.389688,6.9028125)
\usefont{T1}{ptm}{m}{n}
\rput(14.734531,6.8478127){\Large 202}
\usefont{T1}{ptm}{m}{n}
\rput(3.9345312,6.0478125){\Large 232}
\psline[linewidth=0.04cm,arrowsize=0.05291667cm 2.0,arrowlength=1.4,arrowinset=0.4]{<->}(4.3896875,6.1028123)(6.5896873,6.1028123)
\usefont{T1}{ptm}{m}{n}
\rput(6.9246874,6.0478125){\Large 233}
\psline[linewidth=0.04cm,arrowsize=0.05291667cm 2.0,arrowlength=1.4,arrowinset=0.4]{<->}(7.3896875,6.1028123)(9.589687,6.1028123)
\usefont{T1}{ptm}{m}{n}
\rput(9.935938,6.0478125){\Large 230}
\psline[linewidth=0.04cm,arrowsize=0.05291667cm 2.0,arrowlength=1.4,arrowinset=0.4]{<->}(10.389688,6.1028123)(12.589687,6.1028123)
\usefont{T1}{ptm}{m}{n}
\rput(12.915313,6.0478125){\Large 231}
\psline[linewidth=0.04cm,arrowsize=0.05291667cm 2.0,arrowlength=1.4,arrowinset=0.4]{<->}(4.3896875,6.3028126)(5.3896875,6.7028127)
\psline[linewidth=0.04cm,arrowsize=0.05291667cm 2.0,arrowlength=1.4,arrowinset=0.4]{<->}(7.3896875,6.3028126)(8.389688,6.7028127)
\psline[linewidth=0.04cm,arrowsize=0.05291667cm 2.0,arrowlength=1.4,arrowinset=0.4]{<->}(10.389688,6.3028126)(11.389688,6.7028127)
\psline[linewidth=0.04cm,arrowsize=0.05291667cm 2.0,arrowlength=1.4,arrowinset=0.4]{<->}(13.389688,6.3028126)(14.389688,6.7028127)
\usefont{T1}{ptm}{m}{n}
\rput(2.1153126,5.2478123){\Large 221}
\psline[linewidth=0.04cm,arrowsize=0.05291667cm 2.0,arrowlength=1.4,arrowinset=0.4]{<->}(2.5896876,5.3028126)(4.7896876,5.3028126)
\usefont{T1}{ptm}{m}{n}
\rput(5.134531,5.2478123){\Large 222}
\psline[linewidth=0.04cm,arrowsize=0.05291667cm 2.0,arrowlength=1.4,arrowinset=0.4]{<->}(5.5896873,5.3028126)(7.7896876,5.3028126)
\usefont{T1}{ptm}{m}{n}
\rput(8.124687,5.2478123){\Large 223}
\psline[linewidth=0.04cm,arrowsize=0.05291667cm 2.0,arrowlength=1.4,arrowinset=0.4]{<->}(8.589687,5.3028126)(10.789687,5.3028126)
\usefont{T1}{ptm}{m}{n}
\rput(11.135938,5.2478123){\Large 220}
\psline[linewidth=0.04cm,arrowsize=0.05291667cm 2.0,arrowlength=1.4,arrowinset=0.4]{<->}(2.5896876,5.5028124)(3.5896876,5.9028125)
\psline[linewidth=0.04cm,arrowsize=0.05291667cm 2.0,arrowlength=1.4,arrowinset=0.4]{<->}(5.5896873,5.5028124)(6.5896873,5.9028125)
\psline[linewidth=0.04cm,arrowsize=0.05291667cm 2.0,arrowlength=1.4,arrowinset=0.4]{<->}(8.589687,5.5028124)(9.589687,5.9028125)
\psline[linewidth=0.04cm,arrowsize=0.05291667cm 2.0,arrowlength=1.4,arrowinset=0.4]{<->}(11.589687,5.5028124)(12.589687,5.9028125)
\usefont{T1}{ptm}{m}{n}
\rput(0.3359375,4.4478126){\Large 210}
\psline[linewidth=0.04cm,arrowsize=0.05291667cm 2.0,arrowlength=1.4,arrowinset=0.4]{<->}(0.7896875,4.5028124)(2.9896874,4.5028124)
\usefont{T1}{ptm}{m}{n}
\rput(3.3153124,4.4478126){\Large 211}
\psline[linewidth=0.04cm,arrowsize=0.05291667cm 2.0,arrowlength=1.4,arrowinset=0.4]{<->}(3.7896874,4.5028124)(5.9896874,4.5028124)
\usefont{T1}{ptm}{m}{n}
\rput(6.3345313,4.4478126){\Large 212}
\psline[linewidth=0.04cm,arrowsize=0.05291667cm 2.0,arrowlength=1.4,arrowinset=0.4]{<->}(6.7896876,4.5028124)(8.989688,4.5028124)
\usefont{T1}{ptm}{m}{n}
\rput(9.324688,4.4478126){\Large 213}
\psline[linewidth=0.04cm,arrowsize=0.05291667cm 2.0,arrowlength=1.4,arrowinset=0.4]{<->}(0.7896875,4.7028127)(1.7896875,5.1028123)
\psline[linewidth=0.04cm,arrowsize=0.05291667cm 2.0,arrowlength=1.4,arrowinset=0.4]{<->}(3.7896874,4.7028127)(4.7896876,5.1028123)
\psline[linewidth=0.04cm,arrowsize=0.05291667cm 2.0,arrowlength=1.4,arrowinset=0.4]{<->}(6.7896876,4.7028127)(7.7896876,5.1028123)
\psline[linewidth=0.04cm,arrowsize=0.05291667cm 2.0,arrowlength=1.4,arrowinset=0.4]{<->}(9.789687,4.7028127)(10.789687,5.1028123)
\usefont{T1}{ptm}{m}{n}
\rput(5.7115626,3.0478125){\Large 132}
\psline[linewidth=0.04cm,arrowsize=0.05291667cm 2.0,arrowlength=1.4,arrowinset=0.4]{<->}(6.1896877,3.1028125)(8.389688,3.1028125)
\usefont{T1}{ptm}{m}{n}
\rput(8.701718,3.0478125){\Large 133}
\psline[linewidth=0.04cm,arrowsize=0.05291667cm 2.0,arrowlength=1.4,arrowinset=0.4]{<->}(9.189688,3.1028125)(11.389688,3.1028125)
\usefont{T1}{ptm}{m}{n}
\rput(11.712969,3.0478125){\Large 130}
\psline[linewidth=0.04cm,arrowsize=0.05291667cm 2.0,arrowlength=1.4,arrowinset=0.4]{<->}(12.189688,3.1028125)(14.389688,3.1028125)
\usefont{T1}{ptm}{m}{n}
\rput(14.692344,3.0478125){\Large 131}
\usefont{T1}{ptm}{m}{n}
\rput(3.8923438,2.2478125){\Large 121}
\psline[linewidth=0.04cm,arrowsize=0.05291667cm 2.0,arrowlength=1.4,arrowinset=0.4]{<->}(4.3896875,2.3028126)(6.5896873,2.3028126)
\usefont{T1}{ptm}{m}{n}
\rput(6.9115624,2.2478125){\Large 122}
\psline[linewidth=0.04cm,arrowsize=0.05291667cm 2.0,arrowlength=1.4,arrowinset=0.4]{<->}(7.3896875,2.3028126)(9.589687,2.3028126)
\usefont{T1}{ptm}{m}{n}
\rput(9.901719,2.2478125){\Large 123}
\psline[linewidth=0.04cm,arrowsize=0.05291667cm 2.0,arrowlength=1.4,arrowinset=0.4]{<->}(10.389688,2.3028126)(12.589687,2.3028126)
\usefont{T1}{ptm}{m}{n}
\rput(12.912969,2.2478125){\Large 120}
\psline[linewidth=0.04cm,arrowsize=0.05291667cm 2.0,arrowlength=1.4,arrowinset=0.4]{<->}(4.3896875,2.5028124)(5.3896875,2.9028125)
\psline[linewidth=0.04cm,arrowsize=0.05291667cm 2.0,arrowlength=1.4,arrowinset=0.4]{<->}(7.3896875,2.5028124)(8.389688,2.9028125)
\psline[linewidth=0.04cm,arrowsize=0.05291667cm 2.0,arrowlength=1.4,arrowinset=0.4]{<->}(10.389688,2.5028124)(11.389688,2.9028125)
\psline[linewidth=0.04cm,arrowsize=0.05291667cm 2.0,arrowlength=1.4,arrowinset=0.4]{<->}(13.389688,2.5028124)(14.389688,2.9028125)
\usefont{T1}{ptm}{m}{n}
\rput(2.1129687,1.4478126){\Large 110}
\psline[linewidth=0.04cm,arrowsize=0.05291667cm 2.0,arrowlength=1.4,arrowinset=0.4]{<->}(2.5896876,1.5028125)(4.7896876,1.5028125)
\usefont{T1}{ptm}{m}{n}
\rput(5.092344,1.4478126){\Large 111}
\psline[linewidth=0.04cm,arrowsize=0.05291667cm 2.0,arrowlength=1.4,arrowinset=0.4]{<->}(5.5896873,1.5028125)(7.7896876,1.5028125)
\usefont{T1}{ptm}{m}{n}
\rput(8.111563,1.4478126){\Large 112}
\psline[linewidth=0.04cm,arrowsize=0.05291667cm 2.0,arrowlength=1.4,arrowinset=0.4]{<->}(8.589687,1.5028125)(10.789687,1.5028125)
\usefont{T1}{ptm}{m}{n}
\rput(11.101719,1.4478126){\Large 113}
\psline[linewidth=0.04cm,arrowsize=0.05291667cm 2.0,arrowlength=1.4,arrowinset=0.4]{<->}(2.5896876,1.7028126)(3.5896876,2.1028125)
\psline[linewidth=0.04cm,arrowsize=0.05291667cm 2.0,arrowlength=1.4,arrowinset=0.4]{<->}(5.5896873,1.7028126)(6.5896873,2.1028125)
\psline[linewidth=0.04cm,arrowsize=0.05291667cm 2.0,arrowlength=1.4,arrowinset=0.4]{<->}(8.589687,1.7028126)(9.589687,2.1028125)
\psline[linewidth=0.04cm,arrowsize=0.05291667cm 2.0,arrowlength=1.4,arrowinset=0.4]{<->}(11.589687,1.7028126)(12.589687,2.1028125)
\usefont{T1}{ptm}{m}{n}
\rput(0.30171874,0.6478125){\Large 103}
\psline[linewidth=0.04cm,arrowsize=0.05291667cm 2.0,arrowlength=1.4,arrowinset=0.4]{<->}(0.7896875,0.7028125)(2.9896874,0.7028125)
\usefont{T1}{ptm}{m}{n}
\rput(3.3129687,0.6478125){\Large 100}
\psline[linewidth=0.04cm,arrowsize=0.05291667cm 2.0,arrowlength=1.4,arrowinset=0.4]{<->}(3.7896874,0.7028125)(5.9896874,0.7028125)
\usefont{T1}{ptm}{m}{n}
\rput(6.2923436,0.6478125){\Large 101}
\psline[linewidth=0.04cm,arrowsize=0.05291667cm 2.0,arrowlength=1.4,arrowinset=0.4]{<->}(6.7896876,0.7028125)(8.989688,0.7028125)
\usefont{T1}{ptm}{m}{n}
\rput(9.311563,0.6478125){\Large 102}
\psline[linewidth=0.04cm,arrowsize=0.05291667cm 2.0,arrowlength=1.4,arrowinset=0.4]{<->}(0.7896875,0.9028125)(1.7896875,1.3028125)
\psline[linewidth=0.04cm,arrowsize=0.05291667cm 2.0,arrowlength=1.4,arrowinset=0.4]{<->}(3.7896874,0.9028125)(4.7896876,1.3028125)
\psline[linewidth=0.04cm,arrowsize=0.05291667cm 2.0,arrowlength=1.4,arrowinset=0.4]{<->}(6.7896876,0.9028125)(7.7896876,1.3028125)
\psline[linewidth=0.04cm,arrowsize=0.05291667cm 2.0,arrowlength=1.4,arrowinset=0.4]{<->}(9.789687,0.9028125)(10.789687,1.3028125)
\psline[linewidth=0.02cm,arrowsize=0.05291667cm 2.0,arrowlength=1.4,arrowinset=0.4]{<->}(14.789687,6.5028124)(14.789687,3.3028126)
\psline[linewidth=0.02cm,arrowsize=0.05291667cm 2.0,arrowlength=1.4,arrowinset=0.4]{<->}(12.989688,5.7028127)(12.989688,2.5028124)
\psline[linewidth=0.02cm,arrowsize=0.05291667cm 2.0,arrowlength=1.4,arrowinset=0.4]{<->}(11.189688,4.9028125)(11.189688,1.7028126)
\psline[linewidth=0.02cm,arrowsize=0.05291667cm 2.0,arrowlength=1.4,arrowinset=0.4]{<->}(9.389688,4.1028123)(9.389688,0.9028125)
\psline[linewidth=0.02cm,arrowsize=0.05291667cm 2.0,arrowlength=1.4,arrowinset=0.4]{<->}(11.789687,6.5028124)(11.789687,3.3028126)
\psline[linewidth=0.02cm,arrowsize=0.05291667cm 2.0,arrowlength=1.4,arrowinset=0.4]{<->}(9.989688,5.7028127)(9.989688,2.5028124)
\psline[linewidth=0.02cm,arrowsize=0.05291667cm 2.0,arrowlength=1.4,arrowinset=0.4]{<->}(8.189688,4.9028125)(8.189688,1.7028126)
\psline[linewidth=0.02cm,arrowsize=0.05291667cm 2.0,arrowlength=1.4,arrowinset=0.4]{<->}(6.3896875,4.1028123)(6.3896875,0.9028125)
\psline[linewidth=0.02cm,arrowsize=0.05291667cm 2.0,arrowlength=1.4,arrowinset=0.4]{<->}(8.789687,6.5028124)(8.789687,3.3028126)
\psline[linewidth=0.02cm,arrowsize=0.05291667cm 2.0,arrowlength=1.4,arrowinset=0.4]{<->}(6.9896874,5.7028127)(6.9896874,2.5028124)
\psline[linewidth=0.02cm,arrowsize=0.05291667cm 2.0,arrowlength=1.4,arrowinset=0.4]{<->}(5.1896877,4.9028125)(5.1896877,1.7028126)
\psline[linewidth=0.02cm,arrowsize=0.05291667cm 2.0,arrowlength=1.4,arrowinset=0.4]{<->}(3.3896875,4.1028123)(3.3896875,0.9028125)
\psline[linewidth=0.02cm,arrowsize=0.05291667cm 2.0,arrowlength=1.4,arrowinset=0.4]{<->}(5.7896876,6.5028124)(5.7896876,3.3028126)
\psline[linewidth=0.02cm,arrowsize=0.05291667cm 2.0,arrowlength=1.4,arrowinset=0.4]{<->}(3.9896874,5.7028127)(3.9896874,2.5028124)
\psline[linewidth=0.02cm,arrowsize=0.05291667cm 2.0,arrowlength=1.4,arrowinset=0.4]{<->}(2.1896875,4.9028125)(2.1896875,1.7028126)
\psline[linewidth=0.02cm,arrowsize=0.05291667cm 2.0,arrowlength=1.4,arrowinset=0.4]{<->}(0.3896875,4.1028123)(0.3896875,0.9028125)
\usefont{T1}{ptm}{m}{n}
\rput(5.7117186,-0.7521875){\Large 021}
\psline[linewidth=0.04cm,arrowsize=0.05291667cm 2.0,arrowlength=1.4,arrowinset=0.4]{<->}(6.1896877,-0.6971875)(8.389688,-0.6971875)
\usefont{T1}{ptm}{m}{n}
\rput(8.730938,-0.7521875){\Large 022}
\psline[linewidth=0.04cm,arrowsize=0.05291667cm 2.0,arrowlength=1.4,arrowinset=0.4]{<->}(9.189688,-0.6971875)(11.389688,-0.6971875)
\usefont{T1}{ptm}{m}{n}
\rput(11.721094,-0.7521875){\Large 023}
\psline[linewidth=0.04cm,arrowsize=0.05291667cm 2.0,arrowlength=1.4,arrowinset=0.4]{<->}(12.189688,-0.6971875)(14.389688,-0.6971875)
\usefont{T1}{ptm}{m}{n}
\rput(14.732344,-0.7521875){\Large 020}
\usefont{T1}{ptm}{m}{n}
\rput(3.9323437,-1.5521874){\Large 010}
\psline[linewidth=0.04cm,arrowsize=0.05291667cm 2.0,arrowlength=1.4,arrowinset=0.4]{<->}(4.3896875,-1.4971875)(6.5896873,-1.4971875)
\usefont{T1}{ptm}{m}{n}
\rput(6.911719,-1.5521874){\Large 011}
\psline[linewidth=0.04cm,arrowsize=0.05291667cm 2.0,arrowlength=1.4,arrowinset=0.4]{<->}(7.3896875,-1.4971875)(9.589687,-1.4971875)
\usefont{T1}{ptm}{m}{n}
\rput(9.930938,-1.5521874){\Large 012}
\psline[linewidth=0.04cm,arrowsize=0.05291667cm 2.0,arrowlength=1.4,arrowinset=0.4]{<->}(10.389688,-1.4971875)(12.589687,-1.4971875)
\usefont{T1}{ptm}{m}{n}
\rput(12.921094,-1.5521874){\Large 013}
\psline[linewidth=0.04cm,arrowsize=0.05291667cm 2.0,arrowlength=1.4,arrowinset=0.4]{<->}(4.3896875,-1.2971874)(5.3896875,-0.8971875)
\psline[linewidth=0.04cm,arrowsize=0.05291667cm 2.0,arrowlength=1.4,arrowinset=0.4]{<->}(7.3896875,-1.2971874)(8.389688,-0.8971875)
\psline[linewidth=0.04cm,arrowsize=0.05291667cm 2.0,arrowlength=1.4,arrowinset=0.4]{<->}(10.389688,-1.2971874)(11.389688,-0.8971875)
\psline[linewidth=0.04cm,arrowsize=0.05291667cm 2.0,arrowlength=1.4,arrowinset=0.4]{<->}(13.389688,-1.2971874)(14.389688,-0.8971875)
\usefont{T1}{ptm}{m}{n}
\rput(2.1210938,-2.3521874){\Large 003}
\psline[linewidth=0.04cm,arrowsize=0.05291667cm 2.0,arrowlength=1.4,arrowinset=0.4]{<->}(2.5896876,-2.2971876)(4.7896876,-2.2971876)
\usefont{T1}{ptm}{m}{n}
\rput(5.132344,-2.3521874){\Large 000}
\psline[linewidth=0.04cm,arrowsize=0.05291667cm 2.0,arrowlength=1.4,arrowinset=0.4]{<->}(5.5896873,-2.2971876)(7.7896876,-2.2971876)
\usefont{T1}{ptm}{m}{n}
\rput(8.111719,-2.3521874){\Large 001}
\psline[linewidth=0.04cm,arrowsize=0.05291667cm 2.0,arrowlength=1.4,arrowinset=0.4]{<->}(8.589687,-2.2971876)(10.789687,-2.2971876)
\usefont{T1}{ptm}{m}{n}
\rput(11.130938,-2.3521874){\Large 002}
\psline[linewidth=0.04cm,arrowsize=0.05291667cm 2.0,arrowlength=1.4,arrowinset=0.4]{<->}(2.5896876,-2.0971875)(3.5896876,-1.6971875)
\psline[linewidth=0.04cm,arrowsize=0.05291667cm 2.0,arrowlength=1.4,arrowinset=0.4]{<->}(5.5896873,-2.0971875)(6.5896873,-1.6971875)
\psline[linewidth=0.04cm,arrowsize=0.05291667cm 2.0,arrowlength=1.4,arrowinset=0.4]{<->}(8.589687,-2.0971875)(9.589687,-1.6971875)
\psline[linewidth=0.04cm,arrowsize=0.05291667cm 2.0,arrowlength=1.4,arrowinset=0.4]{<->}(11.589687,-2.0971875)(12.589687,-1.6971875)
\usefont{T1}{ptm}{m}{n}
\rput(0.3309375,-3.1521876){\Large 032}
\psline[linewidth=0.04cm,arrowsize=0.05291667cm 2.0,arrowlength=1.4,arrowinset=0.4]{<->}(0.7896875,-3.0971875)(2.9896874,-3.0971875)
\usefont{T1}{ptm}{m}{n}
\rput(3.3210938,-3.1521876){\Large 033}
\psline[linewidth=0.04cm,arrowsize=0.05291667cm 2.0,arrowlength=1.4,arrowinset=0.4]{<->}(3.7896874,-3.0971875)(5.9896874,-3.0971875)
\usefont{T1}{ptm}{m}{n}
\rput(6.3323436,-3.1521876){\Large 030}
\psline[linewidth=0.04cm,arrowsize=0.05291667cm 2.0,arrowlength=1.4,arrowinset=0.4]{<->}(6.7896876,-3.0971875)(8.989688,-3.0971875)
\usefont{T1}{ptm}{m}{n}
\rput(9.311719,-3.1521876){\Large 031}
\psline[linewidth=0.04cm,arrowsize=0.05291667cm 2.0,arrowlength=1.4,arrowinset=0.4]{<->}(0.7896875,-2.8971875)(1.7896875,-2.4971876)
\psline[linewidth=0.04cm,arrowsize=0.05291667cm 2.0,arrowlength=1.4,arrowinset=0.4]{<->}(3.7896874,-2.8971875)(4.7896876,-2.4971876)
\psline[linewidth=0.04cm,arrowsize=0.05291667cm 2.0,arrowlength=1.4,arrowinset=0.4]{<->}(6.7896876,-2.8971875)(7.7896876,-2.4971876)
\psline[linewidth=0.04cm,arrowsize=0.05291667cm 2.0,arrowlength=1.4,arrowinset=0.4]{<->}(9.789687,-2.8971875)(10.789687,-2.4971876)
\psline[linewidth=0.02cm,arrowsize=0.05291667cm 2.0,arrowlength=1.4,arrowinset=0.4]{<->}(14.789687,2.7028124)(14.789687,-0.4971875)
\psline[linewidth=0.02cm,arrowsize=0.05291667cm 2.0,arrowlength=1.4,arrowinset=0.4]{<->}(12.989688,1.9028125)(12.989688,-1.2971874)
\psline[linewidth=0.02cm,arrowsize=0.05291667cm 2.0,arrowlength=1.4,arrowinset=0.4]{<->}(11.189688,1.1028125)(11.189688,-2.0971875)
\psline[linewidth=0.02cm,arrowsize=0.05291667cm 2.0,arrowlength=1.4,arrowinset=0.4]{<->}(9.389688,0.3028125)(9.389688,-2.8971875)
\psline[linewidth=0.02cm,arrowsize=0.05291667cm 2.0,arrowlength=1.4,arrowinset=0.4]{<->}(11.789687,2.7028124)(11.789687,-0.4971875)
\psline[linewidth=0.02cm,arrowsize=0.05291667cm 2.0,arrowlength=1.4,arrowinset=0.4]{<->}(9.989688,1.9028125)(9.989688,-1.2971874)
\psline[linewidth=0.02cm,arrowsize=0.05291667cm 2.0,arrowlength=1.4,arrowinset=0.4]{<->}(8.189688,1.1028125)(8.189688,-2.0971875)
\psline[linewidth=0.02cm,arrowsize=0.05291667cm 2.0,arrowlength=1.4,arrowinset=0.4]{<->}(6.3896875,0.3028125)(6.3896875,-2.8971875)
\psline[linewidth=0.02cm,arrowsize=0.05291667cm 2.0,arrowlength=1.4,arrowinset=0.4]{<->}(8.789687,2.7028124)(8.789687,-0.4971875)
\psline[linewidth=0.02cm,arrowsize=0.05291667cm 2.0,arrowlength=1.4,arrowinset=0.4]{<->}(6.9896874,1.9028125)(6.9896874,-1.2971874)
\psline[linewidth=0.02cm,arrowsize=0.05291667cm 2.0,arrowlength=1.4,arrowinset=0.4]{<->}(5.1896877,1.1028125)(5.1896877,-2.0971875)
\psline[linewidth=0.02cm,arrowsize=0.05291667cm 2.0,arrowlength=1.4,arrowinset=0.4]{<->}(3.3896875,0.3028125)(3.3896875,-2.8971875)
\psline[linewidth=0.02cm,arrowsize=0.05291667cm 2.0,arrowlength=1.4,arrowinset=0.4]{<->}(5.7896876,2.7028124)(5.7896876,-0.4971875)
\psline[linewidth=0.02cm,arrowsize=0.05291667cm 2.0,arrowlength=1.4,arrowinset=0.4]{<->}(3.9896874,1.9028125)(3.9896874,-1.2971874)
\psline[linewidth=0.02cm,arrowsize=0.05291667cm 2.0,arrowlength=1.4,arrowinset=0.4]{<->}(2.1896875,1.1028125)(2.1896875,-2.0971875)
\psline[linewidth=0.02cm,arrowsize=0.05291667cm 2.0,arrowlength=1.4,arrowinset=0.4]{<->}(0.3896875,0.3028125)(0.3896875,-2.8971875)
\usefont{T1}{ptm}{m}{n}
\rput(5.731406,-4.5521874){\Large 310}
\psline[linewidth=0.04cm,arrowsize=0.05291667cm 2.0,arrowlength=1.4,arrowinset=0.4]{<->}(6.1896877,-4.4971876)(8.389688,-4.4971876)
\usefont{T1}{ptm}{m}{n}
\rput(8.710781,-4.5521874){\Large 311}
\psline[linewidth=0.04cm,arrowsize=0.05291667cm 2.0,arrowlength=1.4,arrowinset=0.4]{<->}(9.189688,-4.4971876)(11.389688,-4.4971876)
\usefont{T1}{ptm}{m}{n}
\rput(11.73,-4.5521874){\Large 312}
\psline[linewidth=0.04cm,arrowsize=0.05291667cm 2.0,arrowlength=1.4,arrowinset=0.4]{<->}(12.189688,-4.4971876)(14.389688,-4.4971876)
\usefont{T1}{ptm}{m}{n}
\rput(14.720157,-4.5521874){\Large 313}
\usefont{T1}{ptm}{m}{n}
\rput(3.9201562,-5.3521876){\Large 303}
\psline[linewidth=0.04cm,arrowsize=0.05291667cm 2.0,arrowlength=1.4,arrowinset=0.4]{<->}(4.3896875,-5.2971873)(6.5896873,-5.2971873)
\usefont{T1}{ptm}{m}{n}
\rput(6.931406,-5.3521876){\Large 300}
\psline[linewidth=0.04cm,arrowsize=0.05291667cm 2.0,arrowlength=1.4,arrowinset=0.4]{<->}(7.3896875,-5.2971873)(9.589687,-5.2971873)
\usefont{T1}{ptm}{m}{n}
\rput(9.910781,-5.3521876){\Large 301}
\psline[linewidth=0.04cm,arrowsize=0.05291667cm 2.0,arrowlength=1.4,arrowinset=0.4]{<->}(10.389688,-5.2971873)(12.589687,-5.2971873)
\usefont{T1}{ptm}{m}{n}
\rput(12.93,-5.3521876){\Large 302}
\psline[linewidth=0.04cm,arrowsize=0.05291667cm 2.0,arrowlength=1.4,arrowinset=0.4]{<->}(4.3896875,-5.0971875)(5.3896875,-4.6971874)
\psline[linewidth=0.04cm,arrowsize=0.05291667cm 2.0,arrowlength=1.4,arrowinset=0.4]{<->}(7.3896875,-5.0971875)(8.389688,-4.6971874)
\psline[linewidth=0.04cm,arrowsize=0.05291667cm 2.0,arrowlength=1.4,arrowinset=0.4]{<->}(10.389688,-5.0971875)(11.389688,-4.6971874)
\psline[linewidth=0.04cm,arrowsize=0.05291667cm 2.0,arrowlength=1.4,arrowinset=0.4]{<->}(13.389688,-5.0971875)(14.389688,-4.6971874)
\usefont{T1}{ptm}{m}{n}
\rput(2.13,-6.1521873){\Large 332}
\psline[linewidth=0.04cm,arrowsize=0.05291667cm 2.0,arrowlength=1.4,arrowinset=0.4]{<->}(2.5896876,-6.0971875)(4.7896876,-6.0971875)
\usefont{T1}{ptm}{m}{n}
\rput(5.1201563,-6.1521873){\Large 333}
\psline[linewidth=0.04cm,arrowsize=0.05291667cm 2.0,arrowlength=1.4,arrowinset=0.4]{<->}(5.5896873,-6.0971875)(7.7896876,-6.0971875)
\usefont{T1}{ptm}{m}{n}
\rput(8.131406,-6.1521873){\Large 330}
\psline[linewidth=0.04cm,arrowsize=0.05291667cm 2.0,arrowlength=1.4,arrowinset=0.4]{<->}(8.589687,-6.0971875)(10.789687,-6.0971875)
\usefont{T1}{ptm}{m}{n}
\rput(11.110782,-6.1521873){\Large 331}
\psline[linewidth=0.04cm,arrowsize=0.05291667cm 2.0,arrowlength=1.4,arrowinset=0.4]{<->}(2.5896876,-5.8971877)(3.5896876,-5.4971876)
\psline[linewidth=0.04cm,arrowsize=0.05291667cm 2.0,arrowlength=1.4,arrowinset=0.4]{<->}(5.5896873,-5.8971877)(6.5896873,-5.4971876)
\psline[linewidth=0.04cm,arrowsize=0.05291667cm 2.0,arrowlength=1.4,arrowinset=0.4]{<->}(8.589687,-5.8971877)(9.589687,-5.4971876)
\psline[linewidth=0.04cm,arrowsize=0.05291667cm 2.0,arrowlength=1.4,arrowinset=0.4]{<->}(11.589687,-5.8971877)(12.589687,-5.4971876)
\usefont{T1}{ptm}{m}{n}
\rput(0.31078124,-6.9521875){\Large 321}
\psline[linewidth=0.04cm,arrowsize=0.05291667cm 2.0,arrowlength=1.4,arrowinset=0.4]{<->}(0.7896875,-6.8971877)(2.9896874,-6.8971877)
\usefont{T1}{ptm}{m}{n}
\rput(3.33,-6.9521875){\Large 322}
\psline[linewidth=0.04cm,arrowsize=0.05291667cm 2.0,arrowlength=1.4,arrowinset=0.4]{<->}(3.7896874,-6.8971877)(5.9896874,-6.8971877)
\usefont{T1}{ptm}{m}{n}
\rput(6.320156,-6.9521875){\Large 323}
\psline[linewidth=0.04cm,arrowsize=0.05291667cm 2.0,arrowlength=1.4,arrowinset=0.4]{<->}(6.7896876,-6.8971877)(8.989688,-6.8971877)
\usefont{T1}{ptm}{m}{n}
\rput(9.331407,-6.9521875){\Large 320}
\psline[linewidth=0.04cm,arrowsize=0.05291667cm 2.0,arrowlength=1.4,arrowinset=0.4]{<->}(0.7896875,-6.6971874)(1.7896875,-6.2971873)
\psline[linewidth=0.04cm,arrowsize=0.05291667cm 2.0,arrowlength=1.4,arrowinset=0.4]{<->}(3.7896874,-6.6971874)(4.7896876,-6.2971873)
\psline[linewidth=0.04cm,arrowsize=0.05291667cm 2.0,arrowlength=1.4,arrowinset=0.4]{<->}(6.7896876,-6.6971874)(7.7896876,-6.2971873)
\psline[linewidth=0.04cm,arrowsize=0.05291667cm 2.0,arrowlength=1.4,arrowinset=0.4]{<->}(9.789687,-6.6971874)(10.789687,-6.2971873)
\psline[linewidth=0.02cm,arrowsize=0.05291667cm 2.0,arrowlength=1.4,arrowinset=0.4]{<->}(14.789687,-1.0971875)(14.789687,-4.2971873)
\psline[linewidth=0.02cm,arrowsize=0.05291667cm 2.0,arrowlength=1.4,arrowinset=0.4]{<->}(12.989688,-1.8971875)(12.989688,-5.0971875)
\psline[linewidth=0.02cm,arrowsize=0.05291667cm 2.0,arrowlength=1.4,arrowinset=0.4]{<->}(11.189688,-2.6971874)(11.189688,-5.8971877)
\psline[linewidth=0.02cm,arrowsize=0.05291667cm 2.0,arrowlength=1.4,arrowinset=0.4]{<->}(9.389688,-3.4971876)(9.389688,-6.6971874)
\psline[linewidth=0.02cm,arrowsize=0.05291667cm 2.0,arrowlength=1.4,arrowinset=0.4]{<->}(11.789687,-1.0971875)(11.789687,-4.2971873)
\psline[linewidth=0.02cm,arrowsize=0.05291667cm 2.0,arrowlength=1.4,arrowinset=0.4]{<->}(9.989688,-1.8971875)(9.989688,-5.0971875)
\psline[linewidth=0.02cm,arrowsize=0.05291667cm 2.0,arrowlength=1.4,arrowinset=0.4]{<->}(8.189688,-2.6971874)(8.189688,-5.8971877)
\psline[linewidth=0.02cm,arrowsize=0.05291667cm 2.0,arrowlength=1.4,arrowinset=0.4]{<->}(6.3896875,-3.4971876)(6.3896875,-6.6971874)
\psline[linewidth=0.02cm,arrowsize=0.05291667cm 2.0,arrowlength=1.4,arrowinset=0.4]{<->}(8.789687,-1.0971875)(8.789687,-4.2971873)
\psline[linewidth=0.02cm,arrowsize=0.05291667cm 2.0,arrowlength=1.4,arrowinset=0.4]{<->}(6.9896874,-1.8971875)(6.9896874,-5.0971875)
\psline[linewidth=0.02cm,arrowsize=0.05291667cm 2.0,arrowlength=1.4,arrowinset=0.4]{<->}(5.1896877,-2.6971874)(5.1896877,-5.8971877)
\psline[linewidth=0.02cm,arrowsize=0.05291667cm 2.0,arrowlength=1.4,arrowinset=0.4]{<->}(3.3896875,-3.4971876)(3.3896875,-6.6971874)
\psline[linewidth=0.02cm,arrowsize=0.05291667cm 2.0,arrowlength=1.4,arrowinset=0.4]{<->}(5.7896876,-1.0971875)(5.7896876,-4.2971873)
\psline[linewidth=0.02cm,arrowsize=0.05291667cm 2.0,arrowlength=1.4,arrowinset=0.4]{<->}(3.9896874,-1.8971875)(3.9896874,-5.0971875)
\psline[linewidth=0.02cm,arrowsize=0.05291667cm 2.0,arrowlength=1.4,arrowinset=0.4]{<->}(2.1896875,-2.6971874)(2.1896875,-5.8971877)
\psline[linewidth=0.02cm,arrowsize=0.05291667cm 2.0,arrowlength=1.4,arrowinset=0.4]{<->}(0.3896875,-3.4971876)(0.3896875,-6.6971874)
\psellipse[linewidth=0.04,dimen=outer](5.1796875,-2.3571875)(0.57,0.38)
\psline[linewidth=0.04cm](5.3896875,-4.2971873)(5.9896874,-4.6971874)
\psline[linewidth=0.04cm](5.3896875,-4.6971874)(5.9896874,-4.2971873)
\psline[linewidth=0.04cm](12.589687,-5.0971875)(13.189688,-5.4971876)
\psline[linewidth=0.04cm](12.589687,-5.4971876)(13.189688,-5.0971875)
\psline[linewidth=0.04cm](8.989688,-6.6971874)(9.589687,-7.0971875)
\psline[linewidth=0.04cm](8.989688,-7.0971875)(9.589687,-6.6971874)
\psline[linewidth=0.04cm](10.789687,-5.8971877)(11.389688,-6.2971873)
\psline[linewidth=0.04cm](10.789687,-6.2971873)(11.389688,-5.8971877)
\psline[linewidth=0.04cm](14.389688,-4.2971873)(14.989688,-4.6971874)
\psline[linewidth=0.04cm](14.389688,-4.6971874)(14.989688,-4.2971873)
\psline[linewidth=0.04cm](8.389688,-4.2971873)(8.989688,-4.6971874)
\psline[linewidth=0.04cm](8.389688,-4.6971874)(8.989688,-4.2971873)
\psline[linewidth=0.04cm](11.389688,-4.2971873)(11.989688,-4.6971874)
\psline[linewidth=0.04cm](11.389688,-4.6971874)(11.989688,-4.2971873)
\psline[linewidth=0.04cm](8.989688,-2.8971875)(9.589687,-3.2971876)
\psline[linewidth=0.04cm](8.989688,-3.2971876)(9.589687,-2.8971875)
\psline[linewidth=0.04cm](10.789687,-2.0971875)(11.389688,-2.4971876)
\psline[linewidth=0.04cm](10.789687,-2.4971876)(11.389688,-2.0971875)
\psline[linewidth=0.04cm](12.589687,-1.2971874)(13.189688,-1.6971875)
\psline[linewidth=0.04cm](12.589687,-1.6971875)(13.189688,-1.2971874)
\psline[linewidth=0.04cm](14.389688,-0.4971875)(14.989688,-0.8971875)
\psline[linewidth=0.04cm](14.389688,-0.8971875)(14.989688,-0.4971875)
\psline[linewidth=0.04cm](5.3896875,-0.4971875)(5.9896874,-0.8971875)
\psline[linewidth=0.04cm](5.3896875,-0.8971875)(5.9896874,-0.4971875)
\psline[linewidth=0.04cm](8.389688,-0.4971875)(8.989688,-0.8971875)
\psline[linewidth=0.04cm](8.389688,-0.8971875)(8.989688,-0.4971875)
\psline[linewidth=0.04cm](11.389688,-0.4971875)(11.989688,-0.8971875)
\psline[linewidth=0.04cm](11.389688,-0.8971875)(11.989688,-0.4971875)
\psline[linewidth=0.04cm](5.3896875,3.3028126)(5.9896874,2.9028125)
\psline[linewidth=0.04cm](5.3896875,2.9028125)(5.9896874,3.3028126)
\psline[linewidth=0.04cm](8.389688,3.3028126)(8.989688,2.9028125)
\psline[linewidth=0.04cm](8.389688,2.9028125)(8.989688,3.3028126)
\psline[linewidth=0.04cm](11.389688,3.3028126)(11.989688,2.9028125)
\psline[linewidth=0.04cm](11.389688,2.9028125)(11.989688,3.3028126)
\psline[linewidth=0.04cm](14.389688,3.3028126)(14.989688,2.9028125)
\psline[linewidth=0.04cm](14.389688,2.9028125)(14.989688,3.3028126)
\psline[linewidth=0.04cm](12.589687,2.5028124)(13.189688,2.1028125)
\psline[linewidth=0.04cm](12.589687,2.1028125)(13.189688,2.5028124)
\psline[linewidth=0.04cm](10.789687,1.7028126)(11.389688,1.3028125)
\psline[linewidth=0.04cm](10.789687,1.3028125)(11.389688,1.7028126)
\psline[linewidth=0.04cm](8.989688,0.9028125)(9.589687,0.5028125)
\psline[linewidth=0.04cm](8.989688,0.5028125)(9.589687,0.9028125)
\psline[linewidth=0.04cm](8.989688,4.7028127)(9.589687,4.3028126)
\psline[linewidth=0.04cm](8.989688,4.3028126)(9.589687,4.7028127)
\psline[linewidth=0.04cm](5.3896875,7.1028123)(5.9896874,6.7028127)
\psline[linewidth=0.04cm](5.3896875,6.7028127)(5.9896874,7.1028123)
\psline[linewidth=0.04cm](8.389688,7.1028123)(8.989688,6.7028127)
\psline[linewidth=0.04cm](8.389688,6.7028127)(8.989688,7.1028123)
\psline[linewidth=0.04cm](11.389688,7.1028123)(11.989688,6.7028127)
\psline[linewidth=0.04cm](11.389688,6.7028127)(11.989688,7.1028123)
\psline[linewidth=0.04cm](14.389688,7.1028123)(14.989688,6.7028127)
\psline[linewidth=0.04cm](14.389688,6.7028127)(14.989688,7.1028123)
\psline[linewidth=0.04cm](12.589687,6.3028126)(13.189688,5.9028125)
\psline[linewidth=0.04cm](12.589687,5.9028125)(13.189688,6.3028126)
\psline[linewidth=0.04cm](10.789687,5.5028124)(11.389688,5.1028123)
\psline[linewidth=0.04cm](10.789687,5.1028123)(11.389688,5.5028124)
\end{pspicture} 
}
\caption{The digraph $\mathfrak{G}(3)$}
\label{G3}
\end{figure}
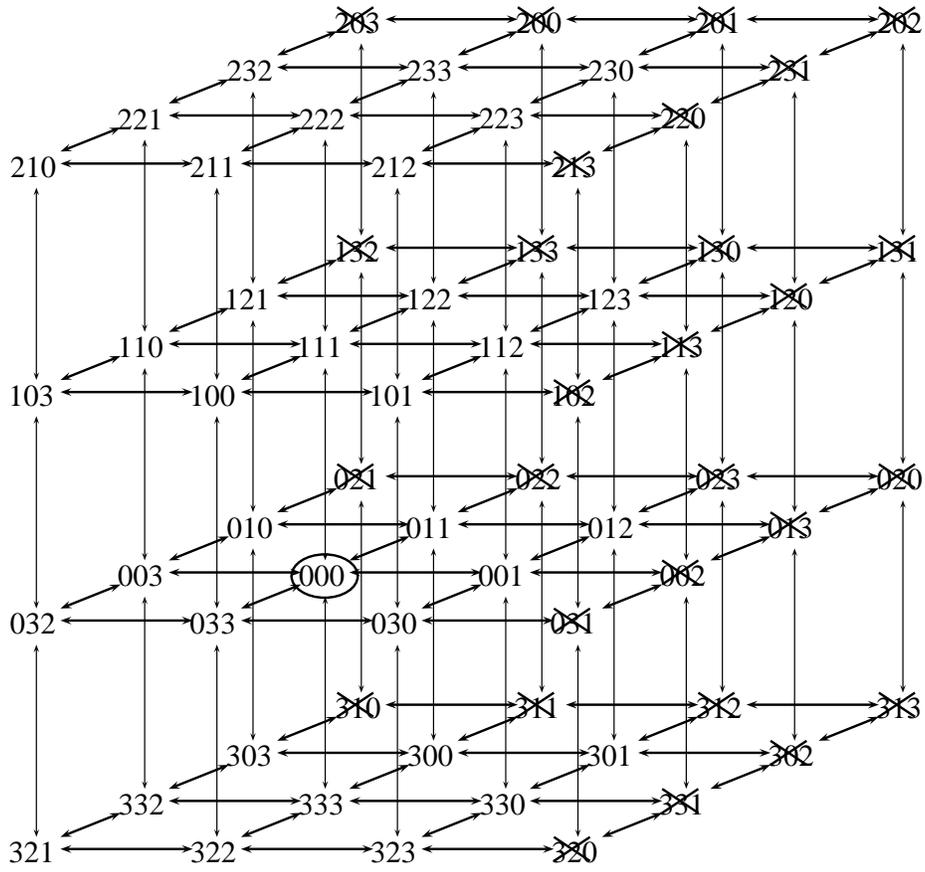

Indeed, being 
able to make one pivot move in a given conformation of size $n$ is equivalent to make
a move from one edge to another adjacent one in the graph $\mathfrak{G}(n)$.
The set of all conformations that are attainable from a given
conformation $c$ by a succession of folding processes are thus exactly
the connected component of $c$. This is why the elements of $SAW_3$ 
are exactly the connected component of the origin $000\hdots 00$.

Furthermore, the program described in Section~\ref{Saw2pasSaw3} is only
able to find connected components reduced to one vertex. Obviously, 
it should be possible to find larger connected components that have
not the origin in their set of connected vertices. 
These vertices are the conformations that are in 
$SAW_2 \setminus SAW_3$. In other words, if $\dot{c}$ is the
connected component of $c$,
\begin{equation}
SAW_2(n) \setminus SAW_3(n) = \left\{c \in \mathfrak{S}(n) \textrm{ s.t. } 000\hdots 00 \notin \dot{c} \right\}
\end{equation}
Such components 
are composed by conformations that can be folded several times, but 
that are not able to be transformed into the line
$0000\hdots 00$. These programs presented previously are thus
only able to determine conformations in the set
\begin{equation}
\left\{c \in \mathfrak{S}(n) \textrm{ s.t. } 000\hdots 00 \notin \dot{c} \textrm{ and cardinality of } \dot{c} \textrm{ is } 1 \right\}
\end{equation}
which is certainly strictly included into $SAW_2(n) \setminus SAW_3(n)$.
The authors' intention is to improve these programs in a future
work, in order to determine if the connected component of a
given vertex contains the origin or not. The problem making 
it difficult to obtain such components is the construction 
of $\mathfrak{S}(n)$. Until now, we:
\begin{itemize}
 \item list the $4^n$ possible walks;
 \item define nodes of the graph from this list, by testing if 
the walk is a path without crossing;
 \item for each node of the graph, we obtain the list of
its $2\times n$ possible neighbors;
 \item an edge between the considered vertex and one of its 
possible neighbors is added if and only if this neighbor is
a path without crossing.
\end{itemize}
Then we compare the size of the connected component of the 
origin to the number of vertices into the graph (this latter is indeed the number of $n$-step self-avoiding walks 
on square lattice as defined in Madras and Slane, that corresponds to the Sloane's A001411 highly non-trival sequence;
it is known that there are $\alpha^n$ self avoiding walks, with upper and lower bounds on the value $\alpha$). If the 
difference is large, then the proof of completeness is
irrelevant. Obviously, our computational 
approach can only provide results for small $n$ 
corresponding to peptides, not proteins.
These results are listed into Table~\ref{composante connexe} 
and the ratio is represented in Figure~\ref{PlotDeLaccroissement}.

\begin{table}
 \centering
 \begin{tabular}{c|c|c|c}
 $n$ & Size of the connected comp. of $00\dots 0$ & Nodes in $\mathfrak{S}(n)$ & Nodes in $\mathfrak{S}_0(n)$\\
\hline
1 & 4 & 4 & 4 \\
2 & 12 & 12 & 16 \\
3 & 36 & 36 & 64 \\
4 & 100 & 100 & 256 \\
5 & 284 & 284 & 1024 \\
6 & 780 & 780 & 4096 \\
7 & 2172 & 2172 & 16384 \\
8 & 5916 & 5916 & 65536 \\
9 & 16268 & 16268 & 262144 \\
10 & 44100 & 44100 & 10485576 \\
11 & 120292 & 120292 & 4194304
 \end{tabular}
\caption{Sizes ratio between $SAW_2(n)$ and $SAW_3(n)$ for small $n$}
\label{composante connexe}
\end{table}

\begin{figure}
 \centering
\includegraphics[scale=0.55]{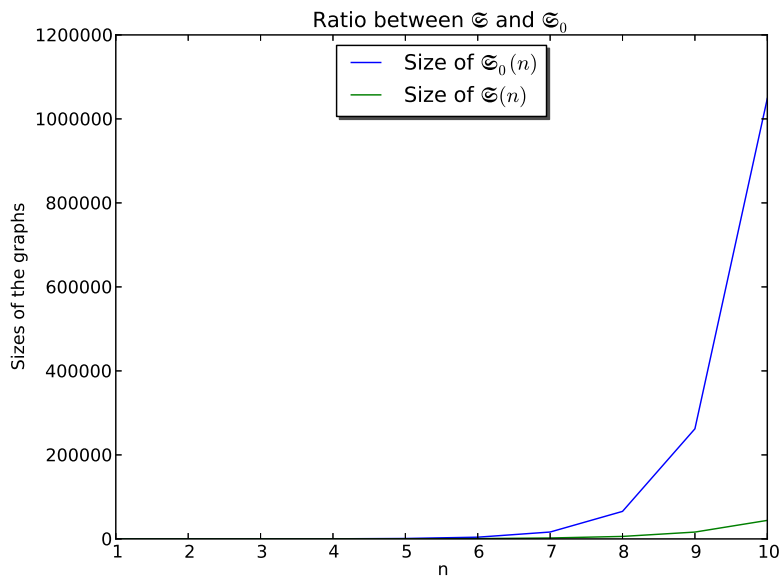}
\caption{Number of nodes removed in $\mathfrak{S}_0(n)$}
\label{PlotDeLaccroissement}
\end{figure}

One can deduce from these results that for small $n$, there is only one 
connected component in $\mathfrak{S}(n)$, and thus $SAW_2(n)=SAW_3(n)$
for $n\leqslant 11$. On the contrary, the previous section shows that 
$SAW_2(n) \neq SAW_3(n)$ for $n$ equal to $158, 169, 175,$ and $914$.
It seems as if a stall appears between $n=11$ and $n=158$ making a rupture
in the connectivity of $\mathfrak{S}(n)$: too much vertices from 
$\mathfrak{S}_0(n)$ have been removed to preserve its connectivity when
defining $\mathfrak{S}(n)$. As the difference between the sizes of
$\mathfrak{S}_0(n)$ and $\mathfrak{S}(n)$ increases more and more, we can
reasonably suppose that the remaining nodes are more and more isolated, 
leading to the appearance of several connected components, and to the
reduction of the size of the component of the origin.

\section{Conclusion}

In this paper, 
the 2D HP square lattice model used for low resolution 
prediction has been investigated. 
We have shown that its SAW requirement can be understood in at least four 
different ways. 
Then we have demonstrated that
these four sets are not equal. In particular,
$SAW_4$ is strictly 
included into $SAW_3$, which is strictly included
into $SAW_1$. 
So the NP-completeness
proof has been realized in a larger set that the
one used for prediction, which is larger than the
set of biologically possible conformations.
Concrete examples have been given, and characterizations of these
sets have finally been proposed.

At this point, we can claim that the NP-completeness of the protein folding
prediction problem does not hold due to the fact that it has been 
established for a set that is not natural in the biological world: it 
encompasses too much conformations as it takes place into $SAW_2$.
However, this discussion still remains qualitative, and if the size of
$SAW_3$ is very large, then the PSP problem is probably an NP-complete
one (even if the proof still remains to do).

We will try to compare in a future work the size of $SAW_2$, which is the Sloane's A001411 sequence, 
to the size of the connected component of the origin.
The third dimension will be investigated, and mathematical results of
the self-avoiding walks belonging into $SAW_3$ will be regarded. 
Conclusion of these studies will then be 
dressed, and solutions to improve the quality of the protein
structure prediction will finally be investigated.

\bibliographystyle{plain}
\bibliography{mabase}

\newpage
\section*{Appendix}

This appendix contains the Python programs that have helped the authors during
their investigations of the respective $SAW_i$.

\subsection{The list of all possible conformations}

Python function called \emph{conformations} (Listing~\ref{AllConformations}) 
produces the list of all possible conformations (satisfying or not 
the excluded volume requirement) as follows: the conformations
of length $n$ are the conformations of size $n-1$ with 0, 1, 2, or
3 added to their tails (recursive call). The return is a list of
conformations, that is, a list of integers lists.

\lstset{caption={Obtaining all the conformations},label=AllConformations}
\begin{lstlisting}
def conformations(n):
    if n==1:
        return [[0]]
    else:
    L = []
    for k in conformations(n-1):
        for i in range(4):
            L.append(k+[i])
    return L
\end{lstlisting}

\subsection{Obtaining the $SAW_1$ conformations}

To obtain the conformations belonging into $SAW_1$, we 
first introduce the function \emph{points} which aim is
to produce the list of points (two coordinates) of the
square lattice that corresponds to a given conformation $C$.
This is simply the function $p$ of Section~\ref{pathWithout}.

Function \emph{is\_saw1} returns a Boolean: it is true if and
only if the conformation $C$ is in $SAW_1$. To do so, the
list of its points in the lattice (its support) is produced, and it is
regarded whether this list contains twice a same point (in
other words, if the support has the same size than 
the list of points).

Finally, \emph{saw1\_conformations} produces a generator. It
returns the next $SAW_1$ conformation at each call of the 
\emph{next} method on the generator. To do so, an exhaustive
iteration of the list produced by \emph{conformations} is
realized, and the \emph{is\_saw1} function is applied to each 
element of this list, to test if this element is in $SAW_1$.

\lstset{caption={Finding the $SAW_1$ elements},label=SAW1Conformations}
\begin{lstlisting}
def points(C):
    L = [(0,0)]
    for c in C:
        P = L[-1]
        if c == 0: L.append((P[0]+1,P[1]))
        elif c == 1: L.append((P[0],P[1]-1))
        elif c == 2: L.append((P[0]-1,P[1]))
        elif c == 3: L.append((P[0],P[1]+1))
     return L


def is_saw1(C):
    L = points(C)
    return len(L) == len(list(set(L)))


def saw1_conformations(n):
    for k in conformations(n):
        if is_saw1(k):
            yield k
\end{lstlisting}

\subsection{Investigating the $SAW_3$ set}

To determine if a conformation in $SAW_2$ is in $SAW_3$ too, we try all the possible 
pivot moves (either in the clockwise direction, or in the anticlockwise). The 
\emph{fold} function tests, considering a conformation called \emph{protein}, a pivot
move on residue number \emph{position} following the given \emph{direction} (+1 or
-1, if clockwise or not). Function \emph{is\_in\_SAW3} applies the \emph{fold} function
to each residue of the candidate, and for the two possible directions. 
The function returns True if and only if no pivot move is possible (the function can
return erroneous False responses for conformations that can be unfolded a few, but
never until the line 0000...00).

\lstset{caption={Testing whether a $SAW_2$ element is in $SAW_3$},label=SAW3Conformations}

\begin{lstlisting}
def fold(protein, position, direction):
    if position == 0:
        return protein
    new_conformation = []
    for k in range(len(protein)):
        if k<abs(position):
            new_conformation.append(protein[k])
        else:
            new_conformation.append((protein[k]+direction)%4)
    return new_conformation


def is_in_SAW3(candidate):
    for k in range(1,len(candidate)):
        if is_saw1(fold(candidate, k, -1)):
            return True
        elif is_saw1(fold(candidate, k, +1)):
            return True
    return False
\end{lstlisting}

\end{document}